\documentclass[journal]{IEEEtran}

\usepackage{amsfonts,amscd,mathrsfs,amsmath,amsthm,amssymb}
\usepackage{mathtools}
\usepackage{xparse}
\usepackage{graphicx}
\usepackage{float}
\usepackage{pifont}
\usepackage[dvipsnames]{xcolor}
\usepackage{colortbl}
\usepackage{multirow}
\usepackage{enumerate}   
\usepackage{comment}
\usepackage{booktabs}
\usepackage{cancel}

\PassOptionsToPackage{hyphens}{url}\usepackage{hyperref}

\hypersetup{
    colorlinks=true,
    linkcolor=magenta!80!black,  
    citecolor=magenta!80!black,
}
\usepackage[noadjust]{cite}

\usepackage[capitalise]{cleveref}
\usepackage{physics}
\usepackage{bm}
\usepackage{bbm}
\usepackage{caption}

\theoremstyle{definition}

\newtheorem{theorem}{Theorem}

\newtheorem*{conjecture*}{Conjecture}

\newtheorem{lemma}{Lemma}

\newtheorem{proposition}{Proposition}

 \newtheorem{code}{Code Family}

\newtheorem{conjcode}{(Conjectured) Code Family}

\newcommand\numberthis{\addtocounter{equation}{1}\tag{\theequation}}

    \newcommand{\bfit}[1]{\textit{#1}}

    \newcommand{\ZZ}{\mathbb{Z}}	%Integers
    \newcommand{\R}{\mathbb{R}} %Reals
    \newcommand{\C}{\mathbb{C}} %Complex
        \newcommand{\Q}{\mathbb{Q}} % Rational

    \DeclareMathAlphabet{\mathsfit}{T1}{\sfdefault}{\mddefault}{\sldefault}
    \SetMathAlphabet{\mathsfit}{bold}{T1}{\sfdefault}{\bfdefault}{\sldefault}

    \newcommand{\X}{\mathsf{X}}
    \newcommand{\Y}{\mathsf{Y}}
    \newcommand{\Z}{\mathsf{Z}}
    \renewcommand{\S}{\mathsf{S}}
    
    \newcommand{\T}{\mathsf{T}}
    
    \newcommand{\Ph}{\mathsf{Ph}}
    \renewcommand{\P}{\mathsf{P}}
    \newcommand{\Glog}{\mathsf{G}}

    \newcommand{\bmsf}[1]{\bm{\mathsf{#1}}}
    \newcommand{\logirr}{\bmsf{\lambda}}

    \newcommand{\SL}{\mathrm{SL}}
    \newcommand{\SU}{\mathrm{SU}}
    \newcommand{\U}{\mathrm{U}}
    \newcommand{\SO}{\mathrm{SO}}
    \renewcommand{\Q}{\mathsf{Q}}

    \newcommand{\D}{\mathscr{D}}
    \newcommand{\BD}{\mathsf{BD}}

    \newcommand{\logzero}{\ket{\overline{0}}}
    \newcommand{\logone}{\ket{\overline{1}}}

    \newcommand{\logicalket}[1]{\ket*{\overline{#1}}}
    \newcommand{\logicalbra}[1]{\bra*{\overline{#1}}}

    \newcommand{\supp}{\text{supp}}

\newcommand{\smallminus}{\text{-}}

\begin{document}

\title{Permutation-Invariant Quantum Codes with Transversal Generalized Phase Gates}

\author{Eric~Kubischta, Ian~Teixeira
%\email{erickub@umd.edu} 
%\email{igt@umd.edu}

\thanks{Both authors contributed equally to this work and are both with the Department
of Mathematics and the Joint Center for Quantum Information and Computer Science, University of Maryland, College Park, Maryland 20742 USA}% <-this % stops a space

}

\maketitle

\begin{abstract}
With respect to the transversal gate group (an invariant of quantum codes), we demonstrate that non-additive codes can outperform stabilizer codes. We do this by constructing spin codes that correspond to permutation-invariant multiqubit codes that can implement generalized phase gates transversally. Of particular note, we construct permutation-invariant quantum codes that implement a transversal $T$ gate using fewer qubits and with a better minimum distance than is possible with the best known stabilizer codes. 
\end{abstract}


\section{Introduction}

A quantum code is usually referred to using three parameters: the number of physical qubits $n$, the dimension of the codespace $K$, and the distance of the code $d$. Stabilizer codes are denoted $ [[n,k,d]] $, where $ K=2^k $ is the dimension of the codespace. More generally, non-stabilizer codes, usually called non-additive codes, are denoted $((n,K,d))$.

To our knowledge, the best non-additive codes have distance that is at most $ 1 $ greater than the best stabilizer codes, for the same $ n $ and $ K $. An example is the non-additive quantum Goethals-Preparata code \cite{QuantumGoethalsPreparataCodes} that encodes $ 217 $ logical qubits into $ 256 $ physical qubits with distance $ 8 $, whereas the best known stabilizer code encodes $ 217 $ logical qubits into $ 256 $ physical qubits with distance $ 7 $ \cite{Grassl:codetables}. This is also supported by the result from \cite{GF4codes,Shadows} that the linear programming bound on distance $d$ for non-additive codes is at most 1 greater than the linear programming bound on distance $d$ for stabilizer codes for all $ n \leq 30 $ and all $k \leq 23$. The limited nature of these improvements has, for the most part, relegated non-additive codes to the sidelines of quantum error correction. But these three parameters are not the only invariants of a quantum code. Although transversal gates may be conjugated by non-entangling gates, the group of logical gates $\Glog$ they form is a code invariant. This fact isn't completely unknown, but few techniques exist that can build codes with a specific transversal gate group $\Glog$. 

This brings us to the Dicke bootstrap (\cref{lem:DickeBootstrap} below). It was noticed in \cite{gross2} that the spin codes constructed in \cite{gross1} could be mapped to permutation-invariant multiqubit codes with similar error correction and transversality properties. The Dicke bootstrap transforms a spin code with logical gate group $\Glog$ and spin distance $d$ into a permutation-invariant multiqubit code with transversal gate group $\Glog$ and distance $d$. Thus, instead of studying multiqubit codes, we can study spin codes. Spin codes are easier to work with because they transform in a single irrep (irreducible representation) of $\SU(2)$ instead of a tensor product of irreps and thus spin codes live in an exponentially smaller Hilbert space than multiqubit codes - they populate only a tiny $ n+1 $ dimensional corner of the full $ 2^n $ dimensional space.

We construct new families of spin codes by studying the relationship between symmetry and the error correction conditions. It is illustrative to focus on the case when the symmetry is given by the order $2^{r+2}$ generalized quaternion group $\Q^{(r)}$, which is the maximal group in the $r$-th level of the single qubit Clifford hierarchy (Appendix B of \cite{anderson2022groups}). We also address a generalization of this group for non powers of two, called the binary dihedral group of even degree, denoted by $ \BD_{2b} $ where $ 2b $ is the degree. Each group $ \BD_{2b} $ is generated by the (determinant 1) Pauli gate $ \X $ and the generalized phase gate $ \sqrt[b]{\Z} $.

The resulting spin codes map to permutation-invariant multiqubit codes \cite{2004permutation,ouyangPI,aydin2023family} under the Dicke bootstrap. We find permutation-invariant multiqubit codes constructed this way that outperform the best known stabilizer codes, either in terms of using less physical qubits $n$, having larger distance $d$, or having more transversal gates. Of all the codes we find, perhaps the most striking is an $((11,2,3))$ code that implements the $ T $ gate transversally. There is a great deal of interest in codes with transversal $ T $  \cite{BravyitransversalT,andersontransversalT,HaahtransversalT,bombintransversalT,KubicatransversalT,smallestT,smallestT2,smallestT3} and the $ ((11,2,3)) $ code we find is likely the smallest possible code with transversal $ T $, belying past claims that the $[[15,1,3]]$ code, which was proven to be the smallest stabilizer code with transversal $ T $ \cite{smallestT,smallestT2,smallestT3}, is minimal. Beyond this, we find transversal $ T $ codes, numerically, with distance $ d=5,7,9,11,13 $ in $ n=27,49,73,107,147 $ qubits respectively. Our codes outperform the best known stabilizer codes with transversal $ T $ \cite{HighDistanceCodesTransversalCliffordTGates}, using fewer qubits to achieve a better distance.

\section{Background}

The Hilbert space of a single qubit is $ \C^2 $, while the Hilbert space of $ n $ qubits is $ (\C^2)^{\otimes n} \cong \C^{2^n} $. So a single qubit \textit{gate} is an element of the unitary group $ \U(2) $, while an $ n $-qubit gate is an element of the unitary group $ \U(2^n) $. An $n$-qubit gate is called \bfit{local} if it can be written as a tensor product of $ n $ single qubit gates. The \bfit{weight} of a local gate
$ \bigotimes_{i=1}^n U_i $ is the number of $ U_i $ that are not (multiples of) the identity matrix. Local gates do not create entanglement, since they act on each qubit separately. Another operation that doesn't create entanglement is swapping qubits (i.e. permuting the factors of $ \C^2 $ in the $ n $-fold tensor product). Thus an $n$-qubit gate is called \bfit{non-entangling} if it can be written as a product of local gates and qubit permutations.

A multiqubit quantum error-correcting code is a $ K $ dimensional subspace, called the \textit{codespace}, of the $ n $-qubit Hilbert space $ (\C^2)^{\otimes n} $. A code is said to have \textit{distance} $ d $ if the action on the codespace of any local gate $ E $ of weight $ \text{wt}(E)<d $ can be detected. This is formalized by the Knill-Laflamme error correction conditions \cite{KL}; a code has distance $ d $ if for all vectors $\logicalket{u}$ and $\logicalket{v}$ in the codespace we have 
\[
\logicalbra{u} E \logicalket{v} = c_E \braket{\overline{u} }{\overline{v}} \qquad \text{for } \text{wt}(E) < d. \numberthis \label{eqn:KL}
\]
The constant $ c_E $ is allowed to depend on $ E $ but not on $\logicalket{u}$ and $\logicalket{v}$. An encoding into $ n $ qubits of a $ K $ dimensional subspace with distance $ d $ is called an $ ((n,K,d)) $ quantum error-correcting code. Since these equations are linear in $ E $ it is equivalent to only check on a basis of Pauli errors of weight less than $ d $ (see Theorem 2 of \cite{gottesman2009introductionquantumerrorcorrection}).

All unitary gates preserve vector space dimension and thus preserve the parameters $n$ and $K$. However only non-entangling gates preserve the weight of a local gate. Thus only non-entangling gates preserve distance $d$. For this reason, two codes are said to be \bfit{equivalent} if they are related by a non-entangling gate \cite{rainsdistance2}.

A local gate that preserves the codespace is called a \bfit{transversal} gate.  Let $ g \in \U(2) $ be a logical gate for an $((n,2,d))$ code. We say that $ g $ is \bfit{exactly transversal} if the physical gate $ g^{\otimes n} $ implements logical $ g $ on the codespace. We say $ g $ is $h$-\bfit{strongly transversal} if there exists some $ h \in \U(2) $, not necessarily equal to $ g $, such that the physical gate $ h^{\otimes n} $  implements logical $ g $ on the codespace. All the codes in this paper have the Pauli gates $ X $ and $ Z $ exactly transversal. And many of the codes, including the $ ((11,2,3)) $ code, have the logical $ T $ gate $ T^3 $-strongly transversal.

\subsection{The Transversal Gate Group}

The group of all logical operations that can be implemented on the codespace using transversal gates is called the \bfit{transversal gate group} $ \Glog $ of the code. In addition to $ n, K, d $, the transversal gate group $\Glog$ is also a well defined code parameter under code equivalence.

\begin{lemma} Equivalent codes have isomorphic transversal gate groups.
\end{lemma}

\begin{proof} Indeed, we prove something even stronger, constructing compatible bases for the two codes, the columns of the encoding maps $ V $ and $ V' $, with respect to which the two transversal gate groups are realized as the exact same matrix group. This result is of course stronger and more constructive than merely proving that the transversal gate groups are isomorphic.  Let $V$ be the encoding map, a $2^n \times K$ matrix whose $ K $ columns correspond to an orthonormal basis of logical states for the codespace. Let $U_L$ be a logical gate (a $K \times K$ matrix) that is implemented transversally on the code by $U_1 \otimes \cdots \otimes U_n$ (a $2^n \times 2^n$ matrix). As an equation this means
\[
   \qty( \bigotimes_i U_i) V = V U_L. \numberthis
\]
The encoding map for an equivalent code is $V' = W V$ where $W$ is some non-entangling gate (a $2^n \times 2^n$ matrix). Then $V = W^\dagger V'$ and we have
\begin{equation}
   \qty( \bigotimes_i  U_i')  V' =  V' U_L. \label{eqn:samegates}
\end{equation}
Here we have defined $\bigotimes_i  U_i' := W\qty( \bigotimes_i  U_i) W^\dagger$ which is indeed still a local gate since we are conjugating by the non-entangling gate $ W $. 

\cref{eqn:samegates} says that an equivalent code still has the same logical gates, they are just implemented by different transversal gates. The claim follows.
\end{proof}

\subsection{Spin Codes}

The unique $ 2j+1 $ dimensional irrep of $ \SU(2) $ is known in physics as a spin $ j $ system (where $j$ is either a non-negative integer or positive half-integer). A basis for the Hilbert space of a spin $ j $ system is given by the kets $\ket{j,m}$, for $m \in \{ -j, 1-j, \dots, j-1, j \}$. The kets $\ket{j,m}$ are an eigenbasis for the action of the diagonal subgroup of $ \SU(2) $ on a spin $ j $ system (in physics, the $\ket{j,m}$ are commonly referred to as the eigenkets for $J_z$, the $z$-component of angular momentum). For an introduction to the ``bra" and ``ket" notation for vectors, as well as an introduction to spin, see Chapter 1 of \cite{sakurai}.

A \bfit{spin code} \cite{gross1} is simply a $K$ dimensional subspace of a spin $ j $ system. The performance of the spin code is measured by how well the subspace protects against products of angular momentum errors $ J_x, J_y, J_z $, which are a basis for the $ 2j+1 $ dimensional irrep of the Lie algebra $ \mathfrak{su}(2)$. Thus we say that a spin code has \bfit{spin distance} $d$ if for all $\logicalket{u}$ and $\logicalket{v}$ in the codespace we have
\[
\logicalbra{u} J_{\alpha_1} \cdots J_{\alpha_p} \logicalket{v} = c \braket{\overline{u} }{\overline{v}} \qquad \text{for } 0 \leq p < d. \numberthis \label{eqn:KLspin}
\]
The constant $c$ is allowed to depend on $\alpha_1, \cdots, \alpha_p$ (here each $ \alpha_i $ takes on the the values $ {x,y,z} $) but not on $\logicalket{u}$ and $\logicalket{v}$. 

The $p$-fold products $J_{\alpha_1} \cdots J_{\alpha_p}$ are Cartesian tensors, and they span a reducible representation of $\SU(2)$, in particular a quotient of the $ p $th tensor power of the adjoint representation of $ \SU(2)$. Naturally we split these representations into irreducible pieces, the \textit{spherical tensors} $T_q^k$ which can be written using the Wigner-Eckart theorem as 
\[
T_q^k = \sqrt{\tfrac{2k+1}{2j+1}} \sum_{m=-j}^j C^{j\, m+q}_{k\,q, j\,m} \dyad{j,m+q}{j,m}   , \numberthis
\]
where $0 \leq |q| \leq k  \leq 2j$, and the $C^{j\,m+q }_{k\,q, j\,m} $ are Clebsch-Gordan coefficients. There are $ (2j+1)^2 $ spherical tensors $T_q^k$ and they are an orthonormal basis  for all linear operators on a spin $j$ Hilbert space, with respect to the trace inner product. More information on spherical tensors can be found in section 3.11 of \cite{sakurai} (and an older reference with different conventions is \cite{quantumtheoryangularmomentum}). 

For any fixed $ k $, the spherical tensors $T_q^k$ with $0 \leq |q| \leq k  $ transform in the irreducible $ 2k+1 $ dimensional irrep of $ \SU(2) $. The $ T_q^k $ for $ 0 \leq |q| \leq k < d $ and the $ J_{\alpha_1} \cdots J_{\alpha_p} $ for $ 0 \leq p < d $ span the same $ d^2 $ dimensional subspace of the operators on a spin $ j $ Hilbert space. The $ T_q^k $ for $ 0 \leq |q| \leq k < d $ are linearly independent and thus a basis for this subspace. However, there are $ \frac{3^d-1}{2} $ of the products $ J_{\alpha_1} \cdots J_{\alpha_p} $ for $ 0 \leq p < d $. Since the inequality $ d^2 \leq \frac{3^d-1}{2} $ is strict for $ d \geq 3 $, these products are in general not linearly independent. So to verify that a spin code has spin distance $ d $, as in \cref{eqn:KLspin}, it is equivalent, and requires far fewer conditions, to check that for all $\logicalket{u}$ and $\logicalket{v}$ in the codespace
\[
    \logicalbra{u} T_q^k \logicalket{v} = c \braket{\overline{u}}{\overline{v}} \quad \text{for } 0 \leq |q| \leq k < d,  \numberthis \label{eqn:KLsphericaltensor}
\]
where the constant $c$ is allowed to depend on $k$ and $q$ but not on $\logicalket{u}$ and $\logicalket{v}$.

\subsection{Dicke Bootstrap}

A spin $j$ system is isomorphic as an $ \SU(2)$ representation to the permutation-invariant subspace of the tensor product of $n = 2j$ spin $1/2$ systems \cite{sakurai}. An explicit isomorphism is the Dicke state mapping
\[
    \ket{j,m} \overset{\D}{\longmapsto} \ket*{D_{j-m}^{2j}}. \numberthis \label{eqn:dicke}
\]
Here $\ket*{D_w^n}$ is a Dicke state \cite{dicke1,dicke2,dicke3,dicke4,dicke5} defined as the (normalized) uniform superposition
    \[
        \ket*{D^n_{w}} = \tfrac{1}{\sqrt{\binom{n}{w}} }\sum_{ wt(s) = w } \ket{s}. \numberthis
    \]
    where the sum is over all $\binom{n}{w}$ of the length $n$ bit strings of Hamming weight $w$. For example,
\[
    \ket*{D_2^3} = \tfrac{1}{\sqrt{3}}\qty( \ket{011} + \ket{101} + \ket{110} ). \numberthis
\]

 Each $g \in \SU(2)$ has a natural action on spin $j$ via the Wigner $ D $ rotation operators $ D^j(g) $. If $D^j(g)$ preserves the codespace then it will implement a logical gate. A spin code is called $ (\Glog,\logirr) $ \bfit{covariant} if $D^j(g)$ implements the logical gate $\logirr(g)$ where $\logirr$ is a representation of $\Glog$ \cite{covariant1,covariant2}. 

\begin{lemma}[Dicke Bootstrap]\label{lem:DickeBootstrap} A $ \Glog $-covariant spin $j$ code with a codespace of dimension $ K $ and a spin distance of $ d $ corresponds under $\D$ to a $\Glog$-transversal $((n,K,d))$ permutation-invariant multiqubit code, where $n = 2j$. 
\end{lemma}

\begin{proof}
 
First we show that a $\Glog$-covariant spin $j$ code corresponds to a $\Glog$-transversal $n$-qubit code. In other words, we need to show that the Dicke state mapping $\D$ behaves as an intertwiner between the natural action of $ \SU(2) $ on a spin $ j $ irrep and the natural action of $ \SU(2) $ on an $ n=2j $ qubit system via the tensor product \cite{us1}:
\[
\D\qty[ D^j(g) \ket{j,m} ] = g^{\otimes n}  \D \ket{j,m} .  \numberthis \label{eqn:dickecovar}
\]
It is actually easier to prove this intertwining equation for $g \in \SL(2,\mathbb{C})$ which implies that it is also true for $g \in \SU(2)$ because $\SU(2)$ is a subgroup of $\SL(2,\mathbb{C})$. Moreover, if we can show this equation is true at the Lie algebra level, then it will automatically be true at the Lie group level via exponentiation. 

The group action on the left hand side of \cref{eqn:dickecovar} can be generated by exponentiating the span of the basis for the $2j+1$ dimensional irrep of $ \mathfrak{sl}(2,\mathbb{C}) $, the Lie algebra of $\SL(2,\mathbb{C})$, where the basis is given by $ J_+, J_-, J_z $ where $J_z$ is the $z$-component of angular momentum and $J_\pm = J_x \pm i J_y$ are ladder operators. The group action on the right hand side of \cref{eqn:dickecovar} can be generated in a corresponding way by exponentiating the span of the collective spin operators $ \sum_{i=1}^n \sigma_\alpha^{(i)} $ where $\sigma_\alpha^{(i)}$ is $\sigma_\alpha$ on the $i$-th qubit, and $ \sigma_\alpha$ is either $\sigma_\pm= \sigma_x \pm i \sigma_y= \tfrac{1}{2}(X \pm i Y) $, or $\sigma_z=\tfrac{1}{2} Z $. Notice that when $ j=\frac{1}{2} $ then $ \D $ is the identity operator and $ J_{\alpha}=\sigma_\alpha $.

Thus in order to show \cref{eqn:dickecovar}, it suffices to show
\[
    \D\qty[ J_\alpha \ket{j,m} ] = \sum_{i=1}^n \sigma_\alpha^{(i)} \D \ket{j,m}. \numberthis \label{eqn:dickebasis}
\]

Let's start with the $\alpha = z$ condition. On the left side of \cref{eqn:dickebasis} we have
\begin{align}
    \D \qty[ J_z \ket{j,m} ] &= m \D \ket{j,m} =m \ket*{D^{2j}_{j-m}}. \label{eqn:dickeZproof}
\end{align}
Here we have used the fact that $J_z \ket{j,m} = m \ket{j,m}$ since $\ket{j,m}$ is by definition an eigenvector of $J_z$ with eigenvalue $m$. Then we have used the Dicke mapping \cref{eqn:dicke}.

On the right hand side of \cref{eqn:dickebasis} we have 
\begin{subequations}
\begin{align}
    \sum_{i=1}^n \sigma_z^{(i)} \D \ket{j,m} &= \tfrac{1}{2}\sum_{i=1}^n Z^{(i)} \ket*{D^{2j}_{j-m}} \\
    &=  \tfrac{1}{2}\sum_{i=1}^n Z^{(i)} \tfrac{1}{\sqrt{\binom{2j}{j-m}} }\sum_{ wt(s) = j-m }  \ket{s} \label{eqn:dickeproof1} \\
    &= \tfrac{1}{\sqrt{\binom{2j}{j-m}} }\sum_{ wt(s) = j-m }  \tfrac{1}{2}\sum_{i=1}^n Z^{(i)} \ket{s} \\
     &= \tfrac{1}{\sqrt{\binom{2j}{j-m}} }\sum_{ wt(s) = j-m }  \tfrac{(j+m)-(j-m)}{2} \ket{s} \label{eqn:dickeproof2} \\
      &= \tfrac{1}{\sqrt{\binom{2j}{j-m}} }\sum_{ wt(s) = j-m }  m \ket{s} \\
    &= m \ket*{D^{2j}_{j-m}}.
\end{align}
\end{subequations}
In \cref{eqn:dickeproof1} we use the definition of the Dicke state. In \cref{eqn:dickeproof2} we use the fact that, since $Z\ket{0} = \ket{0}$, there are  $2j-(j-m)=j+m$ many values of $ i $ such that $ Z^{(i)} \ket{s} =\ket{s}$ and, since $Z\ket{1} = - \ket{1}$, there are $(j-m)$ values of $ i $ such that $ Z^{(i)} \ket{s} =-\ket{s}$. This calculation, together with \cref{eqn:dickeZproof}, proves \cref{eqn:dickebasis} for $J_z$. 

Now let's move on to $\alpha = \pm$. To start, recall \cite{sakurai}
\[
J_\pm \ket{j,m} = \sqrt{(j \mp m)(j\pm m+1)} \ket{j,m\pm 1}. \numberthis \label{eqn:ladders}
\]
Then the left hand side of \cref{eqn:dickebasis} is
\begin{subequations}
\begin{align}
    \D\qty[ J_\pm \ket{j,m}] &= \sqrt{(j \mp m)(j\pm m+1)} \D \ket{j,m\pm 1} \\
    &= \sqrt{(j \mp m)(j\pm m+1)} \ket*{D^{2j}_{j-m \mp 1}}. \label{eqn:dickeproof3}
\end{align}
\end{subequations}

Now let's consider the right hand side of \cref{eqn:dickebasis}. Recall that
$ \sigma_+ = \smqty(0 & 1 \\ 0 & 0)$ and $  \sigma_- = \smqty(0 & 0 \\ 1 & 0) $.
Using the standard convention $\ket{0} = \ket*{\tfrac{1}{2},\tfrac{1}{2}}$ and $\ket{1} = \ket*{\tfrac{1}{2},-\tfrac{1}{2}}$, we have $ \sigma_+ \ket{0} = 0$, $\sigma_+ \ket{1} = \ket{0}$, $\sigma_- \ket{0} = \ket{1}$, and $\sigma_- \ket{1} = 0$. To compute the right hand side of \cref{eqn:dickebasis} we need the following lemma.
\begin{lemma} \label{lem:rightside}
\begin{subequations}
\begin{align}
 \sum_i \sigma^{(i)}_+  \ket*{D^n_w} &= \sqrt{(n-w+1)w} \ket*{D^n_{w-1}}, \\
    \sum_i \sigma^{(i)}_-  \ket*{D^n_w} &= \sqrt{(w+1)(n-w)} \ket*{D^n_{w+1}}.
\end{align}
\end{subequations}
\end{lemma}
\begin{proof}
    Recall that a Dicke state can be expanded as
    \[
        \ket*{D^n_{w}} = \tfrac{1}{\sqrt{\binom{n}{w}} }\sum_{ wt(s) = w } \ket{s}, \numberthis
    \]
    where the sum is over all length $n$ bit strings of Hamming weight $w$.
    Then
    \begin{subequations}
    \begin{align}
        \sum_{i=1}^n \sigma^{(i)}_+  \ket*{D^n_w} &=  \tfrac{1}{\sqrt{\binom{n}{w}} } \sum_{i=1}^n \sum_{ wt(s)=w}  \sigma^{(i)}_+ \ket{s} \\
        &= \tfrac{1}{\sqrt{\binom{n}{w}} } \sum_{i=1}^n \sum_{\substack{wt(s')=w-1 \\ s'_i = 0}  }  \ket{s'} \label{eqn:lem3p1}  \\
        &=  \tfrac{(n-w+1)}{\sqrt{\binom{n}{w}} } \sum_{wt(s')=w-1 }  \ket*{s'} \label{eqn:lem3p2} \\
        &= \sqrt{(n-w+1)w} \ket*{D^n_{w-1}} \label{eqn:lem3p3}.
    \end{align}
    \end{subequations}
    \Cref{eqn:lem3p1} follows from the fact that $ \sigma_+^{(i)} $ annihilates any $ \ket{s} $ with $ s_i=0 $, while if $ s_i=1 $, then $ \sigma_+^{(i)} $ just changes $ s_i $ to a $ 0 $.
     In \cref{eqn:lem3p2}, note that the sum $ \sum_{i=1}^n \sum_{wt(s')=w-1, s'_i = 0  } $ is over exactly $ n \binom{n-1}{w-1} $ terms, and every weight $ w-1 $ bit string  $ s' $ appears in the sum the same number of times. Since there are exactly $ \binom{n}{w-1} $ many weight $ w-1 $ bit strings, and $ n \binom{n-1}{w-1}=(n-w+1) \binom{n}{w-1} $, that accounts for the factor of $(n-w+1)$ in \cref{eqn:lem3p2}. The final line, \cref{eqn:lem3p3}, uses the identity $ \binom{n}{w-1}/\binom{n}{w}=\tfrac{w}{n-w+1} $.

       Similarly,  
       \begin{subequations}
    \begin{align}
        \sum_{i=1}^n \sigma^{(i)}_-  \ket*{D^n_w} &=  \tfrac{1}{\sqrt{\binom{n}{w}} } \sum_{i=1}^n \sum_{ wt(s)=w}  \sigma^{(i)}_- \ket{s} \\
        &= \tfrac{1}{\sqrt{\binom{n}{w}} } \sum_{i=1}^n \sum_{\substack{wt(s')=w+1 \\ s'_i = 1}  }  \ket{s'}  \label{eqn:lem3p4}\\
        &= \tfrac{(w+1)}{\sqrt{\binom{n}{w}} } \sum_{wt(s')=w+1 }  \ket*{s'} \label{eqn:lem3p5} \\
        &= \sqrt{(w+1)(n-w)}\ket*{D^n_{w+1}} \label{eqn:lem3p6}.
    \end{align}
    \end{subequations}
    \Cref{eqn:lem3p4} follows from the fact that $ \sigma_-^{(i)} $ annihilates any $ \ket{s} $ with $ s_i=1 $, while if $ s_i=0 $, then $ \sigma_-^{(i)} $ just changes $ s_i $ to a $ 1 $.
     For \cref{eqn:lem3p5}, note that the sum $ \sum_{i=1}^n \sum_{wt(s')=w+1, s'_i = 1  } $ is over exactly $ n \binom{n-1}{w} $ terms, and every weight $ w+1 $ bit string  $ s' $ appears in the sum the same number of times. Since there are exactly $ \binom{n}{w+1} $ many weight $ w+1 $ bit strings, and $ n \binom{n-1}{w}=(w+1) \binom{n}{w+1} $, that accounts for the factor of $ w+1 $ in \cref{eqn:lem3p5}. The final line, \cref{eqn:lem3p6}, uses the identity $ \binom{n}{w+1}/\binom{n}{w}=\frac{n-w}{w+1} $.

\end{proof}
Plugging in $n = 2j$ and $w = j-m$ to \cref{lem:rightside} then we can compute the right hand side of \cref{eqn:dickebasis} and see it agrees with \cref{eqn:dickeproof3} as desired. This concludes the proof that logical gates are preserved under $\D$.

\ \\
We now prove that a spin code with (spin) distance $d$ maps under the Dicke mapping $ \D $ to a multiqubit code with distance $d$. Let $\ket{u} $ be a spin state and let $\ket{\widetilde{u}} := \D \ket{u}$ be the corresponding permutationally invariant multiqubit state. Then we can rewrite \cref{eqn:dickebasis}  suggestively as
\[
    \D\qty[ J_\alpha \ket{u} ] =  \widetilde{J_\alpha} \ket{\widetilde{u}}, \numberthis
\]
where $ \widetilde{J_{\alpha}}:= \D J_{\alpha} \D^\dagger=  \sum_{i=1}^n \sigma_\alpha^{(i)} $  are the collective spin operators. Assume we have a spin $j$ code with (spin) distance $d$ and recall that the spin distance conditions for spherical tensors given in \cref{eqn:KLsphericaltensor}
are equivalent to the spin distance conditions for products of angular momentum given in \cref{eqn:KLspin}. 
Then for $0 \leq p < d$ we have
\begin{subequations}
\begin{align}
\bra{\widetilde{u}}\widetilde{J_{\alpha_1}} \cdots \widetilde{J_{\alpha_p}} \ket{\widetilde{v}} &= \bra{u} \D^{\dagger} \D J_{\alpha_1} \D^{\dagger} \D  \cdots \D^{\dagger} \D J_{\alpha_p} \D^{\dagger} \D\ket{v} \\
&= \bra{u}  J_{\alpha_1} \cdots  J_{\alpha_p} \ket{v} \\
&= c \braket{u }{v} \\
&= c \bra{u} \D^\dagger \D \ket{v} \\ 
&= c \braket{\widetilde{u}}{\widetilde{v}}.
\end{align} \label{eqn:insertdicke}
\end{subequations}

Now we need to show that
$
\bra{\widetilde{u}}\widetilde{J_{\alpha_1}} \cdots \widetilde{J_{\alpha_p}} \ket{\widetilde{v}}=c \braket{\widetilde{u}}{\widetilde{v}} $ for $ 0 \leq p < d $ implies $ \bra{\widetilde{u}}E \ket{\widetilde{v}}=c_E \braket{\widetilde{u}}{\widetilde{v}} $ for all Pauli errors of weight $ wt(E)<d $, since the latter condition is equivalent to showing that the multiqubit code has distance $ d $ (see Theorem 2 of \cite{gottesman2009introductionquantumerrorcorrection}). 

We proceed by induction. The base case follows from the permutation invariance of the multiqubit states $ \ket{\widetilde{u}},\ket{\widetilde{v}}$ (see \cite{gross2}), we have
\begin{subequations}
\begin{align}
    c \braket{\widetilde{u}}{\widetilde{v}}  &=\bra{\widetilde{u}}  \label{eqn:TildeSpinDistance}\widetilde{J_{\alpha}} \ket{\widetilde{v}} \\
    &=\bra{\widetilde{u}}  \sum_{i=1}^n \sigma_{\alpha}^{(i)} \ket{\widetilde{v}} \\
     &= \sum_{i=1}^n \bra{\widetilde{u}}  \sigma_{\alpha}^{(i)} \ket{\widetilde{v}} \\
     &= \sum_{i=1}^n \bra{\widetilde{u}}  \sigma_{\alpha}^{(i^*)} \ket{\widetilde{v}} \label{eqn:GrossPermutationInvarianceTrickSmall} \\
     &= n \bra{\widetilde{u}}  \sigma_{\alpha}^{(i^*)} \ket{\widetilde{v}}.
\end{align}
\end{subequations}
The first line follows from the spin distance condition for the original spin code, and is just a restatement of \cref{eqn:insertdicke} for the $ p=1 $ case. In \cref{eqn:GrossPermutationInvarianceTrickSmall} we use the permutation invariance of the codewords to permute each index $ i $ to the particular fixed index $ i^* $, using equation (25) from \cite{gross2}. This proves the base case of $p=1$.

Now for the inductive step. Suppose that $ \bra{\widetilde{u}}E \ket{\widetilde{v}}=c_E \braket{\widetilde{u}}{\widetilde{v}} $ for all Pauli errors $ E $ of weight $ wt(E)<p $. Then we need to show that $ \bra{\widetilde{u}}E \ket{\widetilde{v}}=c_E \braket{\widetilde{u}}{\widetilde{v}} $ for all Pauli errors of weight $ wt(E)=p $. The key is to separate the terms with Pauli errors of weight $ p $ and the terms with Pauli errors of weight less than $ p $. We have
\begin{subequations}
\begin{align}
   c \braket{\widetilde{u}}{\widetilde{v}}  &=\bra{\widetilde{u}}  \widetilde{J_{\alpha_1}} \cdots \widetilde{J_{\alpha_p}} \ket{\widetilde{v}} \\
    &=  \bra{\widetilde{u}}  \Bigg(\sum_{i_1=1}^n \sigma_{\alpha_1}^{(i_1)} \Bigg) \cdots \Bigg( \sum_{i_p=1}^n \sigma_{\alpha_p}^{(i_p)} \Bigg) \ket{\widetilde{v}} \\
    &= \sum_{1 \leq i_1, \dots, i_p \leq n} \bra{\widetilde{u}}  \sigma_{\alpha_1}^{(i_1)}  \cdots \sigma_{\alpha_p}^{(i_p)} \ket{\widetilde{v}} \\
    &=  \sum_{ \text{some } wt(E)<p}   \bra{\widetilde{u}} E \ket{\widetilde{v}} \nonumber \\
    &\qquad + \sum_{1 \leq i_1 \neq \dots \neq i_p \leq n} \bra{\widetilde{u}}  \sigma_{\alpha_1}^{(i_1)}  \cdots \sigma_{\alpha_p}^{(i_p)} \ket{\widetilde{v}}  \\
    &=  \sum_{ \text{some } wt(E)<p}  c_E \braket{\widetilde{u}}{\widetilde{v}} \nonumber \\
    &\qquad + \sum_{1 \leq i_1 \neq \dots \neq i_p \leq n} \bra{\widetilde{u}}  \sigma_{\alpha_1}^{(i_1^*)}  \cdots \sigma_{\alpha_p}^{(i_p^*)} \ket{\widetilde{v}} .\label{eqn:GrossPermutationInvarianceTrick}   
\end{align}
\end{subequations}

This holds true for all choices of distinct qubits $ i_1^*, \dots i_p^* $ with $ 0 \leq p < d $. In \cref{eqn:GrossPermutationInvarianceTrick} we use induction on $ p $ to simplify the terms in the first sum and we use the permutation invariance of the codewords to permute each list of distinct qubits $ i_1, \dots i_p $ to the particular list of distinct qubits $ i_1^*, \dots i_p^* $, using exactly equation (25) from \cite{gross2}. Now we can subtract 
 terms over to the other side, yielding

\begin{subequations}
\begin{align}
  &\Big( c-\sum_{ \text{some } wt(E)<p}  c_E \Big) \braket{\widetilde{u}}{\widetilde{v}}  = \frac{n!}{(n-p)!}  \bra{\widetilde{u}}  \sigma_{\alpha_1}^{(i_1^*)}  \cdots \sigma_{\alpha_p}^{(i_p^*)} \ket{\widetilde{v}} \\
  &\frac{(n-p)!}{n!}  \Big( c-\sum_{ \text{some } wt(E)<p}  c_E \Big) \braket{\widetilde{u}}{\widetilde{v}}  = \bra{\widetilde{u}}  \sigma_{\alpha_1}^{(i_1^*)}  \cdots \sigma_{\alpha_p}^{(i_p^*)} \ket{\widetilde{v}}.
\end{align}
\end{subequations}
The term multiplying $ \braket{\widetilde{u}}{\widetilde{v}} $ on the left hand side is some constant independent of $\ket{\widetilde{u}}$ and $\ket{\widetilde{v}}$ and so we can conclude that a spin code with (spin) distance $ d $ is indeed mapped by $ \D $ to a distance $ d $ multiqubit code.

This conclude the proof of the Dicke bootstrap.
\end{proof}

Theorem 3.1 of \cite{automorph}, and its subsequent proof, shows that the only permutation-invariant stabilizer codes are the $ [[k+2,k,2]] $  quantum repetition codes with stabilizer generators $ X^{\otimes (k+2)} $ and $ Z^{\otimes (k+2)}$, for even $ k $. Thus the permutation-invariant multiqubit codes obtained using \cref{lem:DickeBootstrap} will be non-additive in nearly all cases.

\subsection{Binary Dihedral Groups}

We are primarily interested in codes that implement a generalized phase gate $ \sqrt[b]{\Z} $ transversally, so we will restrict to the $K = 2$ case, i.e., the case of encoding a single logical qubit. This means the logical gates come from $\mathrm{U}(2)$. But $\mathrm{U}(2) = e^{i \theta }\SU(2)$, so it suffices to consider gates solely from $\SU(2)$. We will always denote matrices from $\SU(2)$ by sans serif font, see \cref{tab:corr} for our chosen correspondence. 

\begin{table}[htp] 
\centering
 \small  
    \begin{tabular} {l|ll} \toprule  
     & $\U(2)$ & $\SU(2)$ \\  \toprule 
    Pauli-$X$ &  $X = \smqty( 0 & 1 \\ 1 & 0) $ & $\X = \smallminus i X$ \\ \addlinespace
    Pauli-$Y$ & $Y = \smqty( 0 & i \\ \smallminus i & 0)$ & $\Y = \smallminus i Y$ \\ \addlinespace
   Pauli-$Z$ & $Z = \smqty(1 & 0 \\ 0 & \smallminus 1)$ & $\Z = \smallminus i Z$ \\ \midrule  
     Phase & $S = \smqty(1 & 0 \\ 0 & i)$ & $\S = e^{\smallminus i \pi /4} S$ \\ 
       \addlinespace
   $\pi/8$-gate & $ T = \smqty(1 & 0 \\ 0 & e^{i\pi/4} )$ & $\mathsf{T} = \smqty( e^{\smallminus i \pi/8} & 0 \\ 0 & e^{i \pi /8} )$ \\ \addlinespace
   $\alpha$-Phase & $Ph(\alpha) = \smqty( 1 & 0 \\ 0 & e^{i \alpha} )$ & $\Ph(\alpha) = \smqty( e^{\smallminus i \alpha/2} & 0 \\ 0 & e^{i \alpha/2} )$ \\  \bottomrule
    \end{tabular}
    \caption{Traditional gates in $\SU(2)$}
    \label{tab:corr}
\end{table}

The \bfit{binary dihedral} (or ``dicyclic") groups of even degree, $\BD_{2b}$, have order $8b$ and are given by (see Chapter 12 of \cite{dicyclic}): 
\[
     \BD_{2b}  = \expval*{\X, \Ph \qty( \tfrac{2\pi}{2b} ) }. \numberthis
\]
There are also odd degree binary dihedral groups, but we won't consider them here since they don't contain $\Z$ and so there is no good way to pick the codewords to be real, required for \cref{lem:Xcovariance} below.

The group $ \BD_{2b}$ is a lift through the double cover $ \SU(2) \twoheadrightarrow \SO(3) $ of the order $4b$ dihedral subgroup $\mathsf{Dih}_{2b} \subset \SO(3)$. For example, the order 8 group $\BD_2$ is equal to the (determinant $ 1 $) single qubit Pauli group $\P = \expval{\X, \Z}$, and is the lift of the Klein four-group, $\mathsf{Dih}_2$.

In the special case that the degree of $\BD_{2b}$ is a power of two, $2b = 2^r$, the binary dihedral groups are instead referred to as the \bfit{generalized quaternion groups} (again see Chapter 12 of \cite{dicyclic}). We will denote these as
\[
    \Q^{(r)} := \BD_{2^r}. \numberthis
\]
The advantage of this notation is that $\Q^{(r)} \subset \mathsf{C}^{(r)} $, where $ \mathsf{C}^{(r)} $ is the $ r $-th level of the (determinant $ 1 $) single qubit \bfit{Clifford hierarchy}. Recall that the single qubit Clifford hierarchy is defined recursively by
\[
    \mathsf{C}^{(r)} := \{ \mathsf{U} \in \SU(2) : \mathsf{U} \P \mathsf{U}^\dagger \subset \mathsf{C}^{(r-1)} \}, \numberthis
\]
starting from $\mathsf{C}^{(1)}:= \P $, the 1-qubit Pauli group \cite{semiClifford}. In fact, $\Q^{(r)}$ is exactly the $\X$ gate together with the diagonal part of $\mathsf{C}^{(r)}$. In general, $\mathsf{C}^{(r)}$ is not a group, but it was shown in \cite{anderson2022groups} that, for $ r \neq 2 $, $\Q^{(r)}$ is the maximal group contained in $\mathsf{C}^{(r)}$. 

In \cite{us1} we called gates not in the Clifford hierarchy \bfit{exotic}. All binary dihedral groups, except $\Q^{(r)}$, contain exotic gates. Since stabilizer codes cannot have exotic transversal gates \cite{wirthmüller2011automorphisms,disjointness, transuniv}, we have the following result.

\begin{proposition} \label{prop:exoticBDgroups} Any code with transversal gate group $ \Glog$ isomorphic to $\BD_{2b}$, for $b$ not a power of $2$, must be non-additive. 
\end{proposition}
\begin{proof}
The faithful characters of $\BD_{2b}$ all take the value $ e^{\smallminus  2\pi i/4b}+ e^{ 2\pi i/4b} $ (for example, see \cite{GAP}). In \cite{us1} it was shown that a gate in the single qubit Clifford hierarchy must have all its entries in the field $ \mathbb{Q}(e^{2\pi i/2^{r+1}}) $ for some $ r $, and thus have trace valued in the cyclotomic field $ \mathbb{Q}(e^{2\pi i/2^{r+1}}) $. When $ b $ is not a power of $ 2 $, $ e^{\smallminus  2\pi i/4b}+ e^{ 2\pi i/4b} \not \in \mathbb{Q}(e^{2\pi i/2^{r+1}}) $ and thus there cannot exist a faithful representation of $\BD_{2b} $ whose image lies in the Clifford hierarchy.
\end{proof}

For example, $ \BD_6 = \expval{\X,\Ph(\tfrac{2\pi}{6})} $ contains the exotic gate $\Ph(\tfrac{2\pi}{6})$ and thus cannot be the transversal gate group of a stabilizer code. 

In contrast, the transversal gate group of the $[[5,1,3]]$ code contains $\Q^{(1)} = \BD_2 = \expval{\X,\Z}$, the transversal gate group of the $[[7,1,3]]$ Steane code contains $\Q^{(2)} = \BD_4 = \expval{\X, \S}$, and the transversal gate group of the $[[15,1,3]]$ code is $\Q^{(3)} = \BD_8 = \expval{\X, \T}$. In fact, the transversal gate group of the $[[2^{r+1}-1,1,3]]$ code family, for $ r \geq 3 $, is $\Q^{(r)} = \BD_{2^r} $; this follows from Theorem 5 and Example 6 of \cite{disjointness}.

\section{$ \BD_{2b}$ Symmetry and the Error Correction Conditions}

 We begin by reviewing the relationship between Wigner-$D$ matrices and spherical tensors, mostly following the conventions of \cite{sakurai}. The Wigner-$D$ operator $D^j(g)$ has matrix elements $D^j_{m m'}(g) := \bra{j,m} D^j(g) \ket{j,m'}$, where $m$ and $m'$ are in the range $\{ -j, -j+1, \cdots, j-1, j \}$ (and $j$ is a half-integer or an integer). The matrix elements for the gates of interest in this work are 
\begin{subequations}
\begin{align}
    D^j_{m m'}({ \X} ) &=  e^{-i\pi j} \delta_{m,-m'} , \\
    D^j_{m m'}( \Y ) &= e^{-i \pi(j + m) } \delta_{m,-m'}  ,\\
    D^j_{m m'}( \Z ) &= e^{-i \pi m} \delta_{m, m'}  , \\
    D^j_{m m'}( \Ph(\alpha) ) &= e^{-i \alpha m} \delta_{m, m'}.
\end{align}
\end{subequations}
In general
\[
     D^j(g) \ket{j,m'} = \sum_{m} D^j_{m m'}(g) \ket{j,m}, \numberthis
\] 
so we have
\begin{subequations}
\begin{align}
    D^j(\X)\ket{j,m} &= e^{-i\pi j} \ket{j,-m} , \\
    D^j(\Y)\ket{j,m} &= e^{-i\pi (j-m)} \ket{j,-m} , \\
    D^j(\Z)\ket{j,m} &= e^{-i\pi m} \ket{j,m} , \\
    D^j( \Ph(\alpha) )\ket{j,m} &=  e^{-i \alpha m} \ket{j,m}.
\end{align}
\end{subequations}

Spherical tensors are irreducible tensors for $\SU(2)$ and so transform as
\[
    D^j(g)^\dagger T_q^k D^j(g) = \sum_{q'=-k}^k [D^{k }_{q q'}(g)]^* T_{q'}^k \numberthis
\]
where $ * $ denotes complex conjugation and $\dagger$ denotes the Hermitian conjugate of a matrix.
Combining this fact with the properties of the matrix elements yields the following lemma. 

\begin{lemma}[Spherical Tensor Symmetries \cite{sakurai}]\label{lem:sym}
\begin{subequations}
\begin{align}
    D^j(\X)^\dagger  T_q^k D^j(\X) &= (-1)^k T_{-q}^k \\
    D^j(\Y)^\dagger  T_q^k D^j(\Y) &= (-1)^{k+q} T_{-q}^k \\
    D^j(\Z)^\dagger T_q^k D^j(\Z) &= (-1)^q T_q^k \\
    D^j\qty(\Ph\qty( \alpha))^\dagger T_q^k D^j\qty(\Ph\qty( \alpha)) &= e^{  i q \alpha} T_q^k \\
    T_q^{k \dagger} &= (-1)^q T_{-q}^k \\
    T_q^{k *} &= T_q^k.
\end{align}
\end{subequations}
\end{lemma}

\subsection{Support of Spin States}

A spin $ j $ state $\ket{\psi}$ can be written generically with respect to the angular momentum basis as 
\[
    \ket{\psi} = \sum_{m=-j}^j \alpha_m \ket{j,m} \numberthis
\]
for some (complex) coefficients $\alpha_m$. The \bfit{support} of $\ket{\psi}$, denoted $\supp \ket{\psi}$, is the set of those $m$ for which $\alpha_m$ is non-zero.

Consider the group
\[
\BD_{2b}:= \expval{\X, \Ph(\tfrac{2\pi}{2b})}. \numberthis
\]
We define the irrep $\bmsf{\tau}_a$ of $ \BD_{2b} $ for $ 1 \leq a \leq b $ by 
\begin{subequations}
\begin{align}
\bmsf{\tau}_a(\X) &=\X \text{ , }  \\
\bmsf{\tau}_a(\Ph\qty(\tfrac{2\pi}{2b}))& =\Big(\Ph\qty(\tfrac{2\pi}{2b})\Big)^{2a-1}. 
\end{align}
\end{subequations}
For example, $\bmsf{\tau}_1$ is the \textit{fundamental representation} of $\BD_{2b}$ where $\bmsf{\tau}_1(\Ph\qty(\tfrac{2\pi}{2b})) =\Ph\qty(\tfrac{2\pi}{2b})$. The representation $ \bmsf{\tau}_a $ is irreducible, which we can verify by taking the trace and then checking that the character of the representation has norm $ 1 $. Note that the $ \bmsf{\tau}_a $ are not the only irreps of $ \BD_{2b} $ but they do exhaust the two dimensional irreps of $\BD_{2b}$ with Frobenius-Schur indicator $-1$ (also called \textit{quaternionic irreps}) \cite{Serre}. Moreover, note that the irrep $ \bmsf{\tau}_a $ is faithful if and only if $ 2a-1 $ is coprime to $ b $ (thus all the $\bmsf{\tau}_a$ are faithful for the generalized quaternion groups $\Q^{(r)}$).

We choose a basis for a $(\BD_{2b},\bmsf{\tau}_a)$ covariant spin code by taking an eigenbasis of $D^j(\Z) $, and we call the elements of this basis the \textit{codewords}.

\begin{lemma} \label{lem:startinglocations} Define $s:=\tfrac{2a-1}{2}$. A $(\BD_{2b},\bmsf{\tau}_a)$ covariant spin code has codewords with support $ s  + 2b \ZZ $ and $ -s  + 2b \ZZ $.
\end{lemma}
\begin{proof}
The action of $\Ph\qty(\tfrac{2\pi}{2b})$ on the $2j+1$ dimensional irrep of $\SU(2)$ is
\[
 D^j\qty(\Ph\qty(\tfrac{2\pi}{2b})) \ket{j,m} = e^{-i  2\pi m/2b } \ket{j,m}. \numberthis 
\]
 If the spin code is $ (\BD_{2b},\bmsf{\tau}_a) $ covariant, then the codewords must also be eigenvectors of $ D^j(\Ph\qty(\tfrac{2\pi}{2b})) $, so we must restrict the support of each codeword to only those values of $m$ such that $ \tfrac{m}{2b} $ is of the form $ \pm\tfrac{2a-1}{4b} + \ZZ $. Additionally, we know $D^j(\X)\ket{j,m} = e^{-i \pi j} \ket{j,-m}$ which implies, with a slight abuse of notation, $\supp \logone = - \supp \logzero$. This proves the result.
\end{proof}

\subsection{Branching Rules}

We can restrict the $2j+1$ dimensional irrep of $\SU(2)$ to the $ \BD_{2b} $ subgroup to get a $ \BD_{2b} $ representation which will always be reducible for $ j > \tfrac{1}{2}$. Let $\chi_j(g) := \Tr( D^j(g) )$ be the \bfit{character} of the $2j+1$ dimensional irrep of the element $g \in \SU(2)$. Let $\tau_a$ denote the character of the $\BD_{2b}$ irrep $ \bmsf{\tau}_a$. The \textit{multiplicity} of the irrep $ \bmsf{\tau}_a$ in the restricted representation is defined by the formula
\[
   \expval{\tau_a, \chi_j} := \frac{1}{|\BD_{2b}|} \sum_{g \in \BD_{2b}} \tau_a(g)^* \chi_j(g). \numberthis
\]
Colloquially, the multiplicity counts the number of copies of the irrep $\bmsf{\tau}_a$ in spin $j$. The multiplicities of the $ \BD_{2b} $ irreps occurring in the restriction of spin $j$ representations are referred to as the \bfit{branching rules} \cite{branchingrules}, denoted $\SU(2) \downarrow \BD_{2b}$. For more information about character theory see Chapter 2 of \cite{Serre}.

 For example, consider the branching rules $\SU(2) \downarrow \Q^{(3)}$ where $\Q^{(3)} = \BD_8 = \expval{\X, \T}$. One can check that the two-dimensional irreps of $\Q^{(3)}$ have multiplicity $0$ unless $j$ is half-integral. The multiplicities in the half-integral cases are tabulated in \cref{tab:branchingrules}. We observe that the irreps fill up in a ``snaking pattern" from left to right then right to left and back again. Indeed, this snaking pattern is a generic feature for all the $ \bmsf{\tau}_a$ irreps of $\BD_{2b}$. Thus the table generalizes in a fairly straightforward way: replace $8$ by $2b$ and look at the pattern between $\tfrac{1}{2}$ to $2b-\tfrac{1}{2}$.

\begin{table}
    \centering
    \renewcommand{\arraystretch}{1.3}
    \begin{tabular}{c|cccc} \toprule
        Spin $j$ & $\bmsf{\tau}_1$ & $\bmsf{\tau}_2$ & $\bmsf{\tau}_3$ & $\bmsf{\tau}_4$ \\ \midrule
        $8p+\tfrac{1}{2}$ & $2p+1$ & $2p$ & $2p$ & $2p$  \\
        $8p+ \tfrac{3}{2}$ & $2p+1$ & $2p+1$ & $2p$ & $2p$ \\
        $8p + \tfrac{5}{2}$ &  $2p+1$ &  $2p+1$ &  $2p+1$ & $2p$ \\
        $8p + \tfrac{7}{2}$ & $2p+1$ & $2p+1$ & $2p+1$ & $2p+1$ \\
        $8p + \tfrac{9}{2}$ & $2p+1$ & $2p+1$ & $2p+1$ & $2p+2$ \\
        $8p + \tfrac{11}{2}$ & $2p+1$ & $2p+1$ & $2p+2$ & $2p+2$ \\
        $8p + \tfrac{13}{2}$ & $2p+1$ & $2p+2$ & $2p+2$ & $2p+2$ \\
        $8p + \tfrac{15}{2}$ & $2p+2$ & $2p+2$ & $2p+2$ & $2p+2$ \\ \bottomrule
    \end{tabular}
    \caption{Branching rules $\SU(2) \downarrow \Q^{(3)}$ for half-integral spin $j$. Here $p \geq 0$ is an integer. }
    \label{tab:branchingrules}
\end{table}

The multiplicity $ \expval{\tau_a, \chi_j} $ of the irrep $\bmsf{\tau}_a$ in spin $j$ is the number of degrees of freedom one has in choosing a $ \BD_{2b} $ covariant code transforming in the irrep $ \bmsf{\tau}_a $. For example, in spin $j= 11/2$ the $\bmsf{\tau}_3$ irrep of $\Q^{(3)}$ has multiplicity 2. Using \cref{lem:startinglocations} the codeword basis can be written as
\begin{subequations}
    \begin{align}
        \logicalket{0} &=\alpha_1 \ket{\tfrac{11}{2},\tfrac{11}{2}}+\alpha_2 \ket{\tfrac{11}{2},\tfrac{-5}{2}} ,\\ 
        \logicalket{1} &=\alpha_1 \ket{\tfrac{11}{2},\tfrac{-11}{2}}+\alpha_2 \ket{\tfrac{11}{2},\tfrac{5}{2}}.
\end{align}
\end{subequations}

Now consider the following lemma. 
\begin{lemma}\label{lem:firstjgivenmult} The smallest spin $ j $ for a $(\BD_{2b},\bmsf{\tau}_a)$ covariant spin code with $\mu$ degrees of freedom is
\[
\begin{cases}
    j = \mu b + (s-b) & \text{if } \mu \text{ odd} \\
    j = \mu b - s & \text{if } \mu \text{ even}, \numberthis
\end{cases}
\]
where $s=\tfrac{2a-1}{2}$.
\end{lemma}
\begin{proof}
The smallest spin $j$ (corresponding to the smallest dimensional $\SU(2)$ irrep) that branches to $\BD_{2b}$ irreps such that $\bmsf{\tau}_a$ has odd multiplicity $\mu=2p+1$ is spin $j = p(2b)+\tfrac{2a-1}{2}$ (using the generalized version of \cref{tab:branchingrules} we alluded to). Rewriting this in terms of $\mu$ and $s$ yields the first claim.

The smallest spin $j$ (corresponding to the smallest dimensional $\SU(2)$ irrep) that branches to $\BD_{2b}$ irreps such that $\bmsf{\tau}_a$ has even multiplicity $\mu=2p+p$ is spin $j = p(2b) + (2b-\tfrac{2a-1}{2})$. Rewriting this in terms of $\mu$ and $s$ yields the second claim.
\end{proof}

\subsection{A Reduced Set of Error Correction Conditions}

 The spherical tensors $T_q^k$ display simple symmetries under the action of the gates $\X$ and $\Ph(\tfrac{2\pi}{2b})$, see \cref{lem:sym}. We can leverage these symmetries to reduce the number of error conditions we need to check for a $ (\BD_{2b},\bmsf{\tau}_a) $ covariant spin code. To start we have a lemma that makes use of another one of our crucial simplifying assumptions: realness of the codewords.

\begin{lemma}\label{lem:Xcovariance} Suppose a spin code is $\X$-covariant and the codewords $\logicalket{0},\logicalket{1}$ are chosen to be real. Then
\begin{subequations}
\begin{align}
    \logicalbra{0}  T_q^k \logicalket{0} &= (-1)^{q+k} \logicalbra{1}  T_{q}^{k}   \logicalket{1} ,\\
    \logicalbra{0}  T_q^k \logicalket{1} &= (-1)^{q+k} \logicalbra{0}  T_{q}^{k}   \logicalket{1}, \\
    \logicalbra{1}  T_q^k \logicalket{0} &= (-1)^{q+k} \logicalbra{1}  T_{q}^{k}   \logicalket{0}.
\end{align}
\end{subequations}
\end{lemma}
\begin{proof} Let $u, v \in \{0, 1\}$. Then
    \begin{subequations}
    \begin{align}
        \logicalbra{u}  T_q^k \logicalket{v}  &= \logicalbra{u+1} D^{j }(\X)^\dagger  T_q^k  D^j(\X) \logicalket{v+1} \label{eqn:lem7p1} \\
        &= (-1)^k \logicalbra{u+1}  T_{-q}^k  \logicalket{v+1} \label{eqn:lem7p2}\\
        &= (-1)^{q+k} \logicalbra{u+1}  T_{q}^{k\dagger}   \logicalket{v+1} \label{eqn:lem7p3} \\
        &= (-1)^{q+k} \logicalbra{v+1}  T_{q}^{k}   \logicalket{u+1}^* \\
        &= (-1)^{q+k} \logicalbra{v+1}  T_{q}^{k}   \logicalket{u+1}. \label{eqn:lem7p4}
    \end{align}
    \end{subequations}
In \cref{eqn:lem7p1} we use the fact that $\X$ is covariant. In \cref{eqn:lem7p2} and \cref{eqn:lem7p3} we use \cref{lem:sym}. In \cref{eqn:lem7p4} we use the realness of the codewords.
\end{proof}

\begin{theorem} \label{thm:ErrorReduction} Consider a $(\BD_{2b},\bmsf{\tau}_a)$ covariant spin code with real codewords. The error conditions of \cref{eqn:KLsphericaltensor} are equivalent to the following reduced set of conditions:
\begin{itemize}
    \item \textit{On-diagonal}:
    \[
    \logicalbra{0} T_q^k \logicalket{0} = \logicalbra{1} T_q^k \logicalket{1} \quad k < d \text{ odd},\,  q \geq 0,\,  q \equiv 0 \bmod{2b}, \label{eqn:thm1p1} \numberthis
    \]
    \item \textit{Off-diagonal}:
    \[
    \logicalbra{0} T_q^k \logicalket{1} = 0 \quad k < d\text{ odd}, \, q \equiv 2a-1 \bmod{2b}. \label{eqn:thm1p2} \numberthis
    \]
\end{itemize}
\end{theorem}
\begin{proof}
    
To see how symmetry reduces the number of error correction conditions we need to check let's start with the on-diagonal condition; we need our code spanned by $\{ \logicalket{0}, \logicalket{1} \}$ to satisfy
\[
\logicalbra{0} T_q^k \logicalket{0} = \logicalbra{1} T_q^k \logicalket{1}. \numberthis \label{eqn:ondiagtemp}
\]
\cref{lem:Xcovariance} shows that $ \logicalbra{0}  T_q^k \logicalket{0} = (-1)^{q+k} \logicalbra{1}  T_{q}^{k}   \logicalket{1} $ so when $q+k$ is even the on-diagonal condition \cref{eqn:ondiagtemp} is satisfied.
Thus we only need to check the on-diagonal condition when $q+k$ is odd.

Furthermore, \cref{lem:startinglocations} shows that values of $m$ within the same codeword are separated by multiples of $2b$. So if $q$ is not a multiple of $2b$ then, for example, $T_q^k \logicalket{0}$ won't overlap $\logicalket{0}$ (or more technically, the inner product between these two vectors will be 0). This means$ \logicalbra{0} T_q^k \logicalket{0} =0= \logicalbra{1} T_q^k \logicalket{1} $ and the on-diagonal condition \cref{eqn:ondiagtemp} is satisfied. So we only need to check the on-diagonal conditions when $q$ is a multiple of $2b$. In particular this means $q$ is even and so we only need to check the on-diagonal conditions for odd $k$ (since we are only checking $q+k$ odd and $q$ even).

Lastly, suppose $\logicalbra{0} T_q^k \logicalket{0} =  \logicalbra{1} T_q^k \logicalket{1}$ for some $q$ and $k$. If we take the complex conjugate of both sides we have $\logicalbra{0} T_{-q}^k \logicalket{0} =  \logicalbra{1} T_{-q}^k \logicalket{1}$. Thus if $T_q^k$ satisfies the on-diagonal conditions then so does $T_{-q}^k$. This means it is sufficient to check only $q \geq 0$. 

In summary, to verify that the on-diagonal error conditions \cref{eqn:ondiagtemp} are satisfied it is sufficient to check only the case where $k$ is odd, $q \geq 0$, and $q \equiv 0 \bmod{2b}$.

Now consider the off-diagonal condition; we need our code to satisfy
\[
\logicalbra{0}  T_q^k \logicalket{1} = \logicalbra{1}  T_q^k \logicalket{0} = 0. \numberthis \label{eqn:offidagtemp}
\]
\cref{lem:Xcovariance} shows that this condition is satisfied when $q+k$ is odd. Thus we only need to check the off-diagonal conditions when $q+k$ is even. 

If $\logicalbra{0} T_q^k \logicalket{1} = 0$ and we take the conjugate of both sides we get $\logicalbra{1} T_{-q}^k \logicalket{0} = 0$. It is thus sufficient to focus only on the condition $\logicalbra{0} T_q^k \logicalket{1} = 0$. 

Now recall that $T_q^k \ket{j,m}$ is proportional to $\ket{j,m+q}$. Since \cref{lem:startinglocations} specifies that one codeword is supported on $ s + (2b) \ZZ $ and the other codeword is supported on $ -s + (2b) \ZZ $ then a shift of $ q $ will only make the supports overlap if $ q \equiv 2s \mod{2b} $. So $\logicalbra{0} T_q^k \logicalket{1} $ necessarily vanishes unless $q \equiv 2a-1 \mod{2b}$. Thus the only non-trivial off-diagonal errors occur when $q$ is odd. Since $q + k$ is even, it only remains to check the cases for which $k$ is odd. 

In summary, to verify the off-diagonal error conditions \cref{eqn:offidagtemp} it is sufficient to check only the case where $k$ is odd, and $q \equiv 2a-1 \bmod{2b}$. 
\end{proof}

 An immediate implication of this theorem is that spherical tensors $T_q^k$ with even $k$ all automatically satisfy the Knill-Laflamme spin code conditions. In fact this implies that we can always take the distance $d$ to be odd. This is analogous to the multiqubit result we found in \cite{XZexactlytransversal}. For example, all error detecting $d=2$ codes in this work are automatically error correcting $d=3$ codes. 

\section{Finding New Codes}

\subsection{A Family of spin codes}

We will now construct a family of $ \BD_{2b} $ covariant $ d=3 $ spin codes. For a $d = 3$ code we only need to satisfy the error conditions for  spherical tensors with $k=0,1,2$. \cref{thm:ErrorReduction} says that, assuming real codewords, the $ k=0,2 $ errors already satisfy the Knill-Laflamme spin code conditions. The conditions $ \logicalbra{0} T_q^1 \logicalket{1} = 0, \; q=\pm1,0 $ also already satisfy the Knill-Laflamme spin code conditions for the irreps $\bmsf{\tau}_a$ with $ 1< a < b $ ( since $ 1 < 2a-1 < 2b-1 $). So to guarantee minimum distance $ d=3 $, for $ 1< a < b $, it is sufficient to satisfy the condition
\[
\logicalbra{0} T_0^1 \logicalket{0} = \logicalbra{1} T_0^1 \logicalket{1}, \numberthis \label{eqn:ondiag3}
\]
where $T_0^1 \propto \sum_{m=-j}^j m \dyad{j,m}$.

When the codewords have only one degree of freedom the system is over constrained and we will fail to solve the equation because of insufficient variables. If the codewords have 2 degrees of freedom then \cref{lem:startinglocations} allows us to write the codewords as $ \alpha_1 \ket{j,s} + \alpha_2 \ket{j,s-2b} $ and $ \alpha_1 \ket{j,-s} + \alpha_2 \ket{j,2b-s} $.
Then \cref{eqn:ondiag3} reduces to $ s |\alpha_1|^2 + (s-2b) |\alpha_2|^2 =0 $ which has a real solution $ \alpha_1 =  \sqrt{\tfrac{2b-s}{2b}}, \alpha_2 = \sqrt{\tfrac{s}{2b}}$ thus
yielding a $d=3$ spin code for each $j \geq 2b-s$ and each $1 < a < b$. Each code has logical group $\bmsf{\tau}_a(\BD_{2b})=\BD_{2b'}$ where $2b'$ is the order of the generator $ \bmsf{\tau}_a(\Ph \qty( \tfrac{2\pi}{2b} ))=\Ph \qty( \tfrac{2\pi}{2b} )^{2a-1}$, given explicitly by $ 2b'=2b/\mathrm{gcd}(2b,2a-1) $. 

Choosing the irrep $\bmsf{\tau}_{b-1}$, we have $ s=\tfrac{2b-3}{2}$ which yields the smallest value of $ j $ among all of the participating irreps. Since $ \mu=2 $ then by \cref{lem:firstjgivenmult} we can take $j =2b-s= \tfrac{2b+3}{2}$  and the codewords are
\begin{subequations}
\begin{align}
    \logicalket{0} &= \sqrt{ \tfrac{2b-3}{4b}} \ket{ \tfrac{2b+3}{2}, \tfrac{2b+3}{2}} + \sqrt{ \tfrac{2b+3}{4b}}\ket{ \tfrac{2b+3}{2}, \tfrac{-2b+3}{2}}  \\
    \logicalket{1} &= \sqrt{ \tfrac{2b-3}{4b}} \ket{ \tfrac{2b+3}{2}, \tfrac{-2b-3}{2}} + \sqrt{ \tfrac{2b+3}{4b}}\ket{ \tfrac{2b+3}{2}, \tfrac{2b-3}{2}}
\end{align}
\end{subequations}
This is a spin $j=\tfrac{2b+3}{2}$ code with a codespace of dimension $ 2 $ and a spin distance of $ 3 $ with logical group $\BD_{2b'}$. Since we picked $ a=b-1 $, then $ 2b'=2b/\mathrm{gcd}(2b,2b-3) $ and so $ b'=b/3 $ when $3$ divides $b$ and $ b'=b $ otherwise.

\subsection{A Family of Multiqubit Codes}

Now we apply the Dicke bootstrap, yielding a permutation-invariant $((2b+3,2,3))$  multiqubit code with codewords
\begin{subequations}
\begin{align}
    \logzero &=\sqrt{ \tfrac{2b-3}{4b}} \ket{D_0^{2b+3}} + \sqrt{ \tfrac{2b+3}{4b}} \ket{D_{2b}^{2b+3}}  \\
    \logone &=\sqrt{ \tfrac{2b-3}{4b}} \ket{D_{2b+3}^{2b+3}} + \sqrt{ \tfrac{2b+3}{4b}} \ket{D_{3}^{2b+3}}.
\end{align}
\end{subequations}
If $3$ does not divide $b$ then the code implements $\BD_{2b}$ transversally, otherwise it implements $\BD_{2b/3}$ transversally. So for all $ b $ not of the form $ 2^r $ or $ 3(2^r) $ these codes  transversally implement a binary dihedral group other than $\Q^{(r)} = \BD_{2^r}$, a feat which is impossible for any stabilizer code by \cref{prop:exoticBDgroups} (cf. \cite{us1}). 

On the other hand, suppose $2b = 2^r$, then this construction yields a family of codes whose transversal gate group $ \Glog $ is the generalized quaternion group $\Q^{(r)} = \BD_{2^r}$.

\begin{code}\label{code:codefamily1}
    For each $r \geq 3$, there is a $ \Q^{(r)} $-transversal $((2^r+3,2,3))$ permutation-invariant multiqubit code given by
        \begin{subequations}
\begin{align}
    \logzero &=\sqrt{ \tfrac{2^r-3}{2^{r+1}}} \ket{D_0^{2^r+3}} + \sqrt{ \tfrac{2^r+3}{2^{r+1}}} \ket{D_{2^r}^{2^r+3}}  \\
    \logone &=\sqrt{ \tfrac{2^r-3}{2^{r+1}}} \ket{D_{2^r+3}^{2^r+3}} + \sqrt{ \tfrac{2^r+3}{2^{r+1}}} \ket{D_{3}^{2^r+3}}.
\end{align}
\end{subequations}
\end{code}

The $[[2^{r+1}-1,1,3]]$ family of stabilizer codes  are also $ \Q^{(r)} $-transversal, and moreover are optimal among all stabilizer codes, i.e., they are the smallest length $n$ stabilizer codes that have transversal gate group $\Q^{(r)}$ \cite{smallestT}. For example, the $[[15,1,3]]$  code ($r=3$) is the smallest stabilizer code with transversal $T$ and the $[[31,1,3]]$ code ($r=4$) is the smallest stabilizer code with transversal $\sqrt{T}$. In contrast, our code family produces a non-additive $((11,2,3))$ code with transversal $T$ and a non-additive $((19,2,3))$ code with transversal $\sqrt{T}$. In general, for large $r$, our family cuts the number of physical qubits needed approximately in half.

\subsection{Higher Distance Codes}

Finding a $(\BD_{2b},\bmsf{\tau}_a) $ covariant real spin code with distance $d$ amounts to solving a system of $N$ quadratic equations (where $N$ is the number of error conditions from \cref{thm:ErrorReduction}) using codewords that have $\mu$ degrees of freedom. We conjecture that a solution always exists so long as $ \mu > N$. We mark all results based on this conjecture with a triangle $ \vartriangle $.  This conjecture was used successfully in \cite{gross2} to construct many real spin codes covariant for the single qubit Clifford group. For a digression on this conjecture (and why it is reasonable) see the appendix. Thus we will proceed by assuming that a covariant spin code exists whenever $\mu = N + 1$. 

\Cref{lem:firstjgivenmult} then says that the smallest $(\BD_{2b},\bmsf{\tau}_a) $ covariant code appears in spin $j = (N+1)b + (s-b)$ if $\mu = N+1$ is odd and $j = (N+1)b -s $ if $\mu = N+1$ is even. So the smallest such spin $j$ is
\[
\begin{cases}
    j = Nb+s & \text{if } N \text{ even} \\
    j = (N+1)b -s & \text{if } N \text{ odd},
\end{cases} \numberthis \label{eqn:spinguess}
\]
where as always $s = \tfrac{2a-1}{2}$. Then once we apply the Dicke bootstrap we will have a $(\BD_{2b},\bmsf{\tau}_a) $ transversal permutation-invariant code in $n = 2j$ qubits where 
\[
\begin{cases}
    n = 2Nb+2a-1 & \text{if } N \text{ even} \\
    n = 2(N+1)b- 2a +1 & \text{if } N \text{ odd}.
\end{cases} \numberthis \label{eqn:perminvn}
\]

Recall that $N$ is the number of error conditions from \cref{thm:ErrorReduction} and it depends on $b$, $a$, and $d$. We compute a closed form expression for $N$ in the appendix (\cref{thm:Nclosed}). For example, if $b=1$ and $a=1$ (i.e., if we want $\Q^{(1)} = \BD_2$ to be covariant) then we find that there are $N = \tfrac{3}{8}(d^2-1)$ non-trivial error correction conditions (quadratic equations). This particular case reproduces a result found in \cite{aydin2023family}. Plugging this into \cref{eqn:perminvn} we find that the smallest permutation-invariant code with $\Q^{(1)}$ transversal has
\[
    n = \frac{1}{4}(3 d^2 + 1). \numberthis
\]
For example, this predicts that there will be a $((7,2,3))$ code, a $((19,2,5))$ code, and $((37,2,7))$ code: all permutation-invariant and all implementing $\Q^{(1)}$ transversally.

\newcommand{\gc}{\cellcolor{green!10}}
\begin{table*}
    \centering
    \begin{tabular}{ll|ccccccccccc} \toprule 
    & & \multicolumn{11}{c}{$d$} \\  
  & &  1 & 3 & 5 & 7 & 9 & 11 & 13 & $15^{\tiny \vartriangle}$ & $17^{\tiny \vartriangle}$ & $19^{\tiny \vartriangle}$ & $21^{\tiny \vartriangle}$ \\ \midrule 
    $\BD_{2}$ & $\bmsf{\tau}_1$ &  \gc 1 & \gc 7 & \gc 19 & \gc 37 & \gc 61 & \gc 91 & \gc 127 & \gc 169 & \gc 217 & \gc 271 & \gc 331 \\ \midrule 
 \multirow{2}{*}{$\BD_4$} & $\bmsf{\tau}_1$ & \gc 1 & \gc 9 & 23 & \gc 41 & \gc 65 & \gc 97 & 135 & \gc 177 & \gc 225 & \gc 281 & 343 \\
 & $\bmsf{\tau}_2$ &3 & 11 & \gc 21 & 43 & 67 & 99 & \gc 133 & 179 & 227 & 283 & \gc 341 \\ \midrule 
 \multirow{3}{*}{$\BD_6$} & $\bmsf{\tau}_1$ & \gc 1 & \gc 13 & \gc 25 & 47 & \gc 73 & 107 & 143 & 191 & 239 & \gc 289 & 359 \\
 & $\xcancel{\bmsf{\tau}_2}$ & 3 &  9 & 27 & 45 &  69 & 105 & 141 &  183 & 237 & 291 & 351 \\
 & $\bmsf{\tau}_3$ & 5 & 17 & 29 & \gc 43 & 77 & \gc 103 & \gc 139 & \gc 187 & \gc 235 & 293 & \gc 355 \\ \midrule 
 \multirow{4}{*}{$\BD_8$} & $\bmsf{\tau}_1$ & \gc 1 & 17 & 33 & \gc 49 & 79 & 113 & 159 & \gc 193 & \gc 241 & 305 & 369 \\
 & $\bmsf{\tau}_2$  & 3 & 13 & 29 & 51 & 77 & 109 & \gc 147 & 195 & 243 & 301 & 365 \\
 &  $\bmsf{\tau}_3$  & 5 & \gc 11 & \gc 27 & 53 & 75 & \gc 107 & 149 & 197 & 245 & \gc 299 & \gc 363 \\
 &  $\bmsf{\tau}_4$  & 7 & 23 & 39 & 55 & \gc 73 & 119 & 153 & 199 & 247 & 311 & 375 \\ \midrule 
 \multirow{5}{*}{$\BD_{10}$} & $\bmsf{\tau}_1$ & \gc 1 & 21 & 41 & 61 & \gc 81 & 119 & 161 & 219 & 261 & 319 & 379 \\
 & $\bmsf{\tau}_2$ & 3 & 17 & 37 & 57 & 83 & 117 & 157 & \gc 203 & 257 & 317 & 377 \\
 & $\xcancel{\bmsf{\tau}_3}$ & 5 & 15 & 25 & 55 & 85 & 115 & 155 & 195 & 255 & 315 & 375 \\
 & $\bmsf{\tau}_4$ & 7 & \gc 13 & \gc 33 & \gc 53 & 87 & 113 & \gc 153 & 207 & \gc 253 & 313 & 373 \\
 & $\bmsf{\tau}_5$ & 9 & 29 & 49 & 69 & 89 & \gc 111 & 169 & 211 & 269 & \gc 311 & \gc 371 \\ \midrule 
 \multirow{6}{*}{$\BD_{12}$} & $\bmsf{\tau}_1$ & \gc 1 & 25 & 49 & 73 & 97 & \gc 121 & 167 & 217 & 287 & 337 & 407 \\
 & $\xcancel{\bmsf{\tau}_2}$ & 3 & 21 & 45 & 69 & 93 & 123 & 165 & 213 & 267 & 333 & 387 \\
 & $\bmsf{\tau}_3$ & 5 & 19 & \gc 29 & \gc 53 & 91 & 125 & 163 & 211 & 259 & \gc 317 & \gc 389 \\
 & $\bmsf{\tau}_4$ & 7 & \gc 17 & 31 & 55 & \gc 89 & 127 & 161 & \gc 209 & \gc 257 & 319 & 391 \\
 & $\xcancel{\bmsf{\tau}_5}$ & 9 & 15 & 39 & 63 & 87 & 129 & 159 & 207 & 273 & 327 & 393 \\
 & $\bmsf{\tau}_6$ & 11 & 35 & 59 & 83 & 107 & 131 & \gc 157 & 227 & 277 & 347 & 397 \\ \bottomrule
    \end{tabular}
    \caption{Smallest $n$ for a $(\BD_{2b}, \bmsf{\tau}_a)$ transversal permutation-invariant real multiqubit code with distance $d$ (we have crossed out the non-faithful irreps). A cell highlighted green means the code is smallest among all faithful irreps of the given group. We have actually constructed the codes in $1 \leq d \leq 13$ using \textsc{Mathematica} 13.2 up to a numeric threshold of $10^{-12}$. Those codes with a $\vartriangle$ are conjectured to exist but we haven't attempted to find them numerically.}
    \label{fig:allcodes}
\end{table*}

Using \textsc{Mathematica 13.2}, we search for spin codes using our predicted values of $j$ from \cref{eqn:spinguess} by minimizing the conditions in \cref{thm:ErrorReduction} up to a threshold of $10^{-12}$ (see \cite{mathematicacode} for some of the Mathematica code). We have tabulated our results for small values of $b$ and small values of $d$ in \cref{fig:allcodes}. In all cases we have tried, we find the relevant spin code in the exact $j$ given by our prediction (and if we try to find one in a smaller $j$ we fail). This corroborates conjecture $\vartriangle$ for cases with $ n $ as large as $ n=169 $.

Using the Dicke bootstrap, these spin codes correspond to a family of $(\BD_{2b},\bmsf{\tau}_a)$-transversal permutation-invariant multiqubit codes with distance $d$ and
\[
    n \overset{{\tiny \vartriangle}}{=} \tfrac{3}{4}d^2 + \tfrac{(b-1)}{2}d  + \order{1}. \numberthis \label{eqn:codefamilymagnum}
\]
Thus we have a conjectured family of codes that grows quadratically $n \sim \tfrac{3}{4}d^2$ for each $b$ and for each $1 \leq a \leq b$.

The most interesting group to find higher distance codes for is $\Q^{(3)} = \BD_{8}$, since this group contains the $\T$ gate (although the same methods here could be used to postulate a family for any $\BD_{2b}$ group). In particular, we have the following code family.

\stepcounter{conjcode}
\begin{conjcode}[$\vartriangle$]\label{code:codefamily2} Assuming the validity of the conjecture $\vartriangle$ stated at the beginning of this section, there exists, for all odd $ d $,  an $((n,2,d))$ permutation-invariant multiqubit code with transversal gate group $\Q^{(3)} = \expval{\X,\T}$ and with code length given by $n = \tfrac{1}{4}\qty( 3 d^2 + 6d - 7 + 2( d \bmod{8}) )$.

We have numerically constructed these codes for small distances $d=3,5,7,9,11,13$ (see Appendix D). 
\end{conjcode}

All the codes here have transversal $ T $. In particular, if the code transforms in the irrep $ \bmsf{\tau}_a $ then the logical $ T $ gate is $ T^{(2a-1)(-1)^a} $-strongly transversal.

When $d = 3$, then \cref{code:codefamily2} intersects \cref{code:codefamily1}, again yielding an $((11,2,3))$ code with transversal $ T $ (that has exact, not numerical, coefficients). When $d = 5$ we get a $((27,2,5))$ code with transversal $ T $ and when $d = 7$ we get a $((49,2,7))$ code with transversal $ T $. According to the recent paper \cite{HighDistanceCodesTransversalCliffordTGates}, the smallest known $ d=5 $ stabilizer code with transversal $ T $ is the $[[49,1,5]]$ code \cite{distance5transversalT}, which has a lower distance than our $((49,2,7))$ code and uses far more qubits than our $ ((27,2,5)) $ code.

Indeed, for distance $ d=5,7,9,11,13 $, \cite{HighDistanceCodesTransversalCliffordTGates} lists the smallest known stabilizer codes with transversal $ T $ as $[[49,1,5]]$, $[[95,1,7]]$, $[[189,1,9]]$, $[[283,1,11]]$, and $[[441,1,13]]$. We significantly outperform these parameters by numerically constructing the following non-additive codes with transversal $ T $: $((27,2,5))$, $((49,2,7))$, $((73,2,9))$, $((107,2,11))$, and $((147, 2, 13))$.

\section{Future Directions}

The Dicke bootstrap described in this work can be used to translate between spin codes ($\SU(2)$ single-irrep codes) and permutation-invariant multiqubit codes. It would be interesting if this correspondence held for $\SU(q)$ as well. There is a well known correspondence between single irreps of $\SU(q)$ and certain subspaces of multi-qudit spaces of local dimension $q$ \cite{Georgi}, but one would also need to check that the distance conditions are compatible (meaning that an $\SU(q)$ irrep code of distance $d$ maps to a multi-qudit code also of distance $d$). If this is true, then our approach could be useful towards finding quantum qudit codes with generalized phase gates, like the generalized $T$ gate \cite{quditT1,quditT2}. It is not clear that these codes would be smaller than the best known stabilizer codes but it is at least plausible given the results here.

Another approach could be to generalize the Dicke bootstrap to other symmetric subspaces, for example cyclic subspaces. The cyclic subspace is not isomorphic to any single spin, but rather a direct sum of many spins (with possible redundancy). This would make the analysis much more complicated but there would still be a significant reduction in search complexity. The class of cyclic quantum codes is much richer than the class of permutation-invariant quantum codes, for example the only permutation-invariant stabilizer code is the quantum repetition code \cite{automorph}, while many of the most famous stabilizer codes, like the $ [[5,1,3]] $ code, are cyclic (also see \cite{quantumcyclic1,quantumcyclic2}). So it is possible this approach could even lead to a more thorough understanding of stabilizer codes with transversal generalized phase gates.

\section{Conclusion}

Non-additive quantum codes have been mostly overlooked since they provide only meager improvements in the three traditional code parameters $n$, $K$, and $d$. However, accounting for the transversal gate group $ \Glog $, non-additive codes can outperform the best known stabilizer analogs. For example, we have constructed a $\Q^{(r)}$-transversal family of non-additive codes with distance $d=3$ in $n = 2^r+3$ qubits that is smaller in length $n$ than the best known analogous family of stabilizer codes, which have length $n = 2^{r+1} - 1$. We have also constructed $\BD_{2b}$-transversal non-additive codes (for $b$ not a power of $2$) that cannot be realized as stabilizer codes.  Lastly, we have constructed non-additive codes with transversal $ T $ that have parameters $((11,2,3))$, $((27,2,5))$, $((49,2,7))$, $((73,2,9))$, $((107,2,11)$, and $((147,2,13))$, these codes outperform the best known stabilizer codes with transversal $ T $, which have parameters $[[15,1,3]]$, $[[49,1,5]]$, $[[95,1,7]]$, $[[189,1,9]]$, $[[283,1,11]]$, and $[[441,1,13]]$ \cite{HighDistanceCodesTransversalCliffordTGates}. 

All our codes are constructed from spin codes using the Dicke bootstrap. The Dicke bootstrap holds significant untapped promise (see the possible generalizations mentioned in the future directions section) and is the first code construction method specifically designed around the transversal gate group $ \Glog $, an important code parameter for fault tolerance.

\section*{Acknowledgments} 

We thank Sivaprasad Omanakuttan and Jonathan A. Gross for helpful comments regarding their previous work. We thank Victor Albert for helpful conversations regarding the Dicke Bootstrap. We thank Michael Gullans for helpful conversations regarding the transversal gate group. This research was supported by NSF QLCI grant OMA-2120757. This research was also supported in part by the MathQuantum RTG through the NSF RTG grant DMS-2231533. All numeric calculations were carried out with \textsc{Mathematica} 13.2. 
 
\bibliographystyle{IEEEtran}
\bibliography{biblio.bib}{}

\newpage 
\appendices

\makeatletter
\renewcommand{\theequation}{A\arabic{equation}}
\renewcommand{\thetable}{A\arabic{table}}
\renewcommand{\thefigure}{A\arabic{figure}}
\renewcommand{\thelemma}{A\arabic{lemma}}
\renewcommand{\thetheorem}{A\arabic{theorem}}
\setcounter{table}{0}
\setcounter{figure}{0}
\setcounter{lemma}{0}
\setcounter{theorem}{0}
\setcounter{equation}{0}

\section{Counting Errors}

Let $ N $ be the size of the reduced set of error correction conditions given in \cref{thm:ErrorReduction}. In this appendix we compute a closed form for $ N $. Let $ N^\text{on-diag}$ denote the number of on-diagonal conditions and let $ N^\text{off-diag}$ denote the number of off-diagonal conditions.

\begin{lemma} \label{lem:numberofequations}
    \begin{align}
        N^\text{on-diag} &=\sum_{\substack{k=1 \\ k \text{ odd}}}^{d-2} 1 + \lfloor  \tfrac{k}{2b} \rfloor, \\
        N^\text{off-diag} &=\sum_{\substack{k=1 \\ k \text{ odd}}}^{d-2}  1 + \lfloor \tfrac{k-2s}{2b} \rfloor + \lfloor \tfrac{k+2s}{2b} \rfloor.
    \end{align}

\end{lemma}
\begin{proof}
Start with the on-diagonal piece. When $k$ is even, all on-diagonal errors are automatically satisfied. When $k$ is odd, the number of on-diagonal errors that need to be checked is equal to the number of solutions of the modular equation 
    \[
        q \equiv 0 \mod{2b} \qquad \text{for } 0 \leq q \leq k. \numberthis
    \]
    There are $1 + \lfloor  k/2b \rfloor $ such solutions. 

  To guarantee a code has distance $d$, we need to check all $ k $ in the range $0 \leq k \leq d-1$. Even $k$ are automatically satisfied so we can take $d$ to be odd - meaning we only need to sum to $d-2$.

    Now consider the off-diagonal piece. When $k$ is even, all off-diagonal errors are automatically satisfied. When $k$ is odd, the number of off-diagonal errors that need to be checked is equal to the number of solutions to the modular equation 
\[
q \equiv 2s \mod{2b} \qquad \text{for } -k \leq q \leq k, \numberthis
\]
where $2s = 2a-1$ is odd. When $q \geq 2s$, then there are $1 + \lfloor (k-2s)/2b \rfloor$ solutions and when $q \leq 2s$ there are $1 + \lfloor (k+2s)/2b \rfloor$ solutions. We can combine these and subtract 1 to account for the fact we double counted at $ 2s$. Thus the total number of non trivial off-diagonal errors for each value of $k$ is
\[
    1 + \lfloor \tfrac{k-2s}{2b} \rfloor + \lfloor \tfrac{k+2s}{2b} \rfloor . \numberthis
\]
To guarantee distance $d$, we must sum these terms for all odd $ k $ between $1 $ and $ d-1$. 

\end{proof}

We can refine this more as follows. We will use the notation $[x]_y := x \bmod{y}$ for brevity and clarity. 
\begin{lemma} Let $d = 2t+1$ then
    \[
        N^{\text{on-diag}} = \tfrac{t^2}{2b} + \tfrac{t}{2} + \zeta^\text{on-diag},
    \]
    where $\zeta^\text{on-diag} = \tfrac{1}{2b}[t]_b( [t]_b - 2[t-\tfrac{1}{2} ]_b + b -1 )$ is crudely estimated as $0 \leq \zeta^\text{on-diag} \leq b$ for any $t$.
\end{lemma}
\begin{proof} Write $d = 2t+1$. We want to estimate
\[
N^\text{on-diag} =\sum_{\substack{k=1 \\ k \text{ odd}}}^{2t-1} 1 + \lfloor  \tfrac{k}{2b} \rfloor = t + \qty( \sum_{\substack{k=1 \\ k \text{ odd}}}^{2t-1} \lfloor  \tfrac{k}{2b} \rfloor). \numberthis
\]
In the sum write  $k = 2k' +1$. Then
    \begin{align}
        \sum_{\substack{k=1 \\ k \text{ odd}}}^{2t-1 } \lfloor  \tfrac{k}{2b} \rfloor &= \sum_{k' = 0}^{ t-1  } \lfloor  \tfrac{2k' + 1}{2b} \rfloor = \sum_{k' = 0}^{ t-1  } \lfloor  \tfrac{k' + \tfrac{1}{2}}{b} \rfloor.
    \end{align}
    Notice there are $t$ many terms in the sum. If $0 \leq k' \leq b-1$ then the summand is $0$, if $b \leq k' \leq 2b-1$ then the summand is $1$, if $2b \leq k' \leq 3b-1$ then the summand is $2$, etc. Each of these groupings contains a sum over $b$ elements. 
    
    If $t$ was a multiple of $b$, say $t = p b$, then there would be $p$  groupings - labeled $k' = 0$ to $k' = p-1$ - with each grouping containing exactly $b$ elements that each evaluate to $k'$. In other words, the sum would be
    \[
        b \sum_{k'=0}^{p-1} k' = \tfrac{b p (p-1)}{2}. \numberthis
    \]
    More generally, $t$ might not be a multiple of $b$, say $t = pb + \gamma$ where $p = \lfloor \tfrac{t}{b} \rfloor$ and $\gamma = t \bmod{b} = [t]_b$. Then we get the previous with $p$ replaced by  $\lfloor \tfrac{t}{b} \rfloor$. This main piece is
    \[
    \mathcal{M} =\tfrac{b \lfloor \tfrac{t}{b} \rfloor \qty( \lfloor \tfrac{t}{b} \rfloor -1)}{2}. \numberthis
    \]
    But we need to also add in the ``tail" (which comes from the fact that the $t$ terms in the sum cannot be split evenly). The tail has $\gamma = [t]_b$ terms in it all evaluated as the endpoint $\lfloor \tfrac{2t-1}{2b} \rfloor$, i.e., the tail is 
    \[
       \mathcal{T} = [t]_b \lfloor \tfrac{2t-1}{2b} \rfloor . \numberthis
    \]
    Thus the sum is
    \begin{align}
        \mathcal{M} + \mathcal{T} 
        = \tfrac{t^2}{2b} - \tfrac{t}{2} + \tfrac{1}{2b}[t]_b( [t]_b - 2[t-\tfrac{1}{2} ]_b + b -1 ).
    \end{align}
    In the last line we have used the fact that $\lfloor \tfrac{x}{y} \rfloor = \tfrac{x}{y} - \tfrac{1}{y} [x]_y$. This yields the form of $N^\text{on-diag}$ given in the claim. The crude bound is found by noting that $0 \leq [t]_b < b$.
\end{proof}

\begin{lemma} Let $d = 2t+1 $ then 
    \[
         N^\text{off-diag} = \tfrac{t^2}{b} + \tfrac{t}{b} + \tfrac{2a(a-b-1)+b+1}{2b} + \zeta^\text{off-diag}
    \]
    where $\zeta^\text{off-diag}$ is given in \cref{eqn:modmess} and is crudely estimated as $0 \leq \zeta^\text{off-diag} \leq 2b$ for all $t$ and for all $a$. 
\end{lemma}
\begin{proof}
Let $h$ be an odd integer (representing either $2s = 2a-1$ or $-2s = -2a+1$). Recall that $1 \leq a \leq b$ so that $|h| \leq 2b -1 $. Consider the sum
\begin{align}
    \mathcal{S} = \sum_{\substack{k=1 \\ k \text{ odd}}}^{2t-1}   \lfloor \tfrac{k+h}{2b} \rfloor = \sum_{k' = 0 }^{t-1}   \lfloor \tfrac{2k'+1+h}{2b} \rfloor = \sum_{k' = 0 }^{t-1}   \lfloor \tfrac{k'+ h'}{b} \rfloor.
\end{align}
We have defined $k' = (k-1)/2$ and $h' = (1+h)/2$ to be integral. Notice that $-b+1 \leq h' \leq b$.

There are $t$ many terms in the sum. If $0\leq k' + h' < b$ the summand is $0$, if $b \leq k' + h' < 2b$ the summand is $1$, if $2b \leq k' + h' < 3b$ the summand $2$, etc. We will break up what is going on into cases. 

\textit{(Case 1)} When $h' = 0$ we have the exact same results as the previous lemma for $\mathcal{M}$ but now the tail is evaluated at $\lfloor \tfrac{t-1}{b} \rfloor$ (but still contains $[t]_b$ terms). Thus the sum is
\[
    \mathcal{S} = \tfrac{b \lfloor \tfrac{t}{b} \rfloor \qty( \lfloor \tfrac{t}{b} \rfloor -1)}{2} + [t]_b \lfloor \tfrac{t-1}{b} \rfloor. \numberthis
\]

\textit{(Case 2)} When $h' = b$ then the main piece from before starts at $1$ instead of $0$, i.e., the main piece comes from $b \sum_{k'=1}^{p} k' = bp (p+1)/2$ where $p = \lfloor \tfrac{t}{b} \rfloor$. The tail is evaluated to $\lfloor \tfrac{t-1 + b}{b} \rfloor$ so that the sum is
\[
    \mathcal{S} = \tfrac{b \lfloor \tfrac{t}{b} \rfloor \qty( \lfloor \tfrac{t}{b} \rfloor + 1)}{2}  + [t]_b \lfloor \tfrac{t+ b - 1}{b} \rfloor. \numberthis
\]

\textit{(Case 3)} When $1 \leq h' \leq b-1$ then the very first grouping has only $b-h'$ elements in it (not $b$ elements). The summand for each of these is $0$. Thus we really only have a sum of $t-(b-h')$ elements. So the main piece is similar but with $b \sum_{k'=1}^p k' = bp(p+1)/2$ where $p = \lfloor \tfrac{t-(b-h')}{b} \rfloor$. The tail will now have $[t-(b-h')]_b$ terms in it all evaluated to the endpoint $\lfloor \tfrac{t-1 + h'}{b} \rfloor$. Thus
\[
    \mathcal{S} = \tfrac{b \lfloor \tfrac{t-(b-h')}{b} \rfloor \qty(  \lfloor \tfrac{t-(b-h')}{b} \rfloor + 1)}{2} + [t-(b-h')]_b \lfloor \tfrac{t-1+ h'}{b} \rfloor. \numberthis
\]

\textit{(Case 4)} When $-b+1 \leq h' \leq -1$ then the very first grouping has $|h'|$ many elements in it, each evaluated to $-1$, i.e., the sum over the first is grouping is just $h'$. Thus we only have a sum over $t+h'$ elements. So the main piece is $b \sum_{k'=0}^{p-1} k' = bp(p-1)/2$ where $p = \lfloor \tfrac{t+h'}{b} \rfloor$. The tail has $[t+h']_b$ elements each evaluated to  $\lfloor \tfrac{t-1 + h'}{b} \rfloor$. Thus
\[
    \mathcal{S} = h' + \tfrac{b \lfloor \tfrac{t+h'}{b} \rfloor \qty(  \lfloor \tfrac{t+h'}{b} \rfloor - 1)}{2} + [t+h']_b \lfloor \tfrac{t-1+ h'}{b} \rfloor. \numberthis
\]
Now we are interested in 
\begin{align}
    N^\text{on-diag}  &=\sum_{\substack{k=1 \\ k \text{ odd}}}^{d-2}  1 + \lfloor \tfrac{k-2s}{2b} \rfloor + \lfloor \tfrac{k+2s}{2b} \rfloor \\
    &= t + \qty( \sum_{k'=0}^{t-1} \lfloor  \tfrac{2k'+1-(2a-1)}{2b} \rfloor + \sum_{k'=0}^{t-1} \lfloor  \tfrac{2k'+1+(2a-1)}{2b} \rfloor) \\
    &= t + \qty( \sum_{k'=0}^{t-1} \lfloor  \tfrac{k'-a+1}{b} \rfloor + \sum_{k'=0}^{t-1} \lfloor  \tfrac{k'+a}{b} \rfloor).
\end{align}
The term in parenthesis, say $\mathcal{S}'$, can be evaluated using the previous as $\mathcal{S}' = \mathcal{S}_{h'=-a+1} + \mathcal{S}_{h'=a}$. 

When $a = 1$ we use case 1 and case 3. Then we have
\[
    N^\text{off-diag} = \tfrac{t^2}{b} + \tfrac{t}{b} + \tfrac{1-b}{2b} + \zeta^\text{off-diag}, \numberthis
\]
where
\begin{align*}
    \zeta^\text{off-diag} &= \tfrac{1}{2b}[t+1]_b^2 \\
    & + \tfrac{1}{2b} [t+1]_b \qty(b-2-2[t]_b ) \\
    & + \tfrac{1}{2b}[t]_b \qty(b-2 - 2 [t-1]_b) . \numberthis
\end{align*}
Note that $0 \leq \zeta^\text{off-diag} \leq 2b$ is an easy (but crude) estimate. When $a = b$ then we use case 4 and case 2. But one can show we get the exact same thing as for $a = 1$.

This leaves us with $2 \leq a \leq b-1$ which come from case 3 and case 4. In these cases 
\[
    N^\text{off-diag} = \tfrac{t^2}{b} + \tfrac{t}{b} + \tfrac{2a(a-b-1)+b+1}{2b} + \zeta^\text{off-diag}, \numberthis
\]
where
\begin{align*}
    \zeta^\text{off-diag} &= \tfrac{1}{2b}[t+1-a]_b^2 \\
    & + \tfrac{1}{2b}[t+1-a]_b\qty(b-2-2[t-a]_b ) \\
    & + \tfrac{1}{2b}[t+a]_b\qty(b-2-2[-1+a+t]_b + [a+t]_b ). \label{eqn:modmess}\numberthis
\end{align*}
This is also crudely estimated as $0 \leq \zeta^\text{off-diag} \leq 2b$. One can show that $a = a_0$ and $a = b-a_0 + 1$ give the same $N^\text{off-diag}$. One can actually check that this reproduces the $a=1$ and $a=b$ cases so this is the general formula. 
\end{proof}

Let $N$ be the number of quadratic equations from \cref{thm:ErrorReduction}, i.e., $N = N^\text{on-diag} + N^\text{off-diag}$. Recall that $d = 2t+1$, i.e., $t$ is the number of errors we can correct. 
\begin{theorem}\label{thm:Nclosed}
\[
    N = \tfrac{3}{2b} t^2 + \tfrac{(2+b)}{2b} t + \tfrac{2 a (a-b-1)+b+1}{2 b} + \zeta, \numberthis
\]
where
\begin{align*}
    \zeta &= \tfrac{1}{2b} \bigg( ((-a+t+1) \bmod b)^2+\\
    & + (-2 ((t-a) \bmod b)+b-2) ((-a+t+1) \bmod b) \\
    & +((a+t) \bmod
   b) (-2 ((a+t-1) \bmod b)\\
   &+((a+t) \bmod b)+b-2)+(t \bmod b)^2 \\
   &+\left(-2
   ((t-\tfrac{1}{2}) \bmod b)+b-1\right) (t \bmod b) \bigg).  \numberthis
\end{align*}
    A crude estimate for $\zeta$ is $0 \leq \zeta \leq 3b$ for all $t$ and all $a$.
\end{theorem}
Notice that $a = a_0$ and $a = b+1-a_0$ yield the same $N$. For example, $\bmsf{\tau}_1$ and $\bmsf{\tau}_b$ yield the same $N$, $\bmsf{\tau}_2$ and $\bmsf{\tau}_{b-1}$ yield the same $N$, etc. Also notice that for $\Q^{(1)} = \BD_2$ (i.e., $b = 1$ and $a = 1$) we have $\zeta = 0$ and so $N = \tfrac{3}{2}t(t+1)$. This special case reproduces the result found in \cite{aydin2023family}.

\section{Existence of solutions}

Finding a $(\BD_{2b},\bmsf{\tau}_a) $ covariant real spin code with distance $d$ amounts to solving a homogeneous system of $N$ quadratic equations in $\mu$ variables. Note that we could consider the more general case of complex spin codes but then the error correction conditions would be complex sesquilinear forms rather than complex quadratic forms. And while there are powerful tools from algebraic geometry for analyzing systems of complex quadratic forms, such as Groebner bases, these generally do not carry over to sesquilinear forms.

If this were a homogeneous 
system of real \textit{linear} equations then we could guarantee a non-zero real solution exists so long as $\mu > N $, the smallest $ \mu $ for which a solution exists being $\mu = N + 1$. Unfortunately, there is no current theory of existence of real solutions to systems of real quadratic equations (the main obstruction is that $\R$ is not algebraically closed).

Existence results analogous to the linear case do exist for quadratic equations over the complex numbers. Consider a system of algebraically independent quadratic equations given by $x^T B_i x = 0$. Then Bezout's theorem from algebraic geometry guarantees that a non-zero complex solution exists as long as the number of variables, $ \mu $, is strictly greater than the number of equations, $ N $.  So, just as in the linear case, a non-zero complex solution exists when $\mu = N + 1$.  Bezout's theorem is not directly applicable in our case because even if the coefficients of the quadratic forms are all real, it only guarantees the existence of a complex solution, not a real one (which our analysis requires, see \cref{lem:Xcovariance}).

In the present setup we have yet to discover, either numerically or analytically, a situation where the heuristic $\mu = N+1$ fails. This observation was also noticed in \cite{gross2} and the authors used this idea to construct many spin codes covariant for the single qubit Clifford group. Thus we conjecture that whenever we have $ N $ error conditions and $ \mu=N+1 $ degrees of freedom we can find the desired spin code.

\section{Encoding and Decoding}

There isn't a general theory of encoding and decoding generic non-additive codes. However, permutation-invariant non-additive codes (like the ones studied here) have significantly more structure. The theory and experimental realization of the efficient (and fault tolerant) preparation of Dicke states and the encoding of a qubit into symmetric states is well studied \cite{PIencode1,PIencode2,PIencode3,PIencode4,PIencode5,PIencode6}. So the encoding step of our permutation-invariant codes seems plausible. 

In contrast, it was only recently discovered how to decode a permutation-invariant code efficiently \cite{PIdecode}. The idea of the decoding algorithm is to get an error syndrome (in the form of a standard Young tableau) by first measuring the total angular momentum on subsets of the physical qubits. Then one can get back to the original code by projecting within the error subspace and applying a unitary (using quantum Schur transforms or teleportation). The authors of \cite{PIdecode} give protocols to implement such a decoding algorithm on a near term quantum device. 

However, it is not clear whether the previous decoding algorithm is fault-tolerant, and certainly more work needs to be done before we can determine how competitive our codes are to analogous stabilizer codes in practice. Our work in this paper should be primarily seen as motivation to study the underdeveloped theory of non-additive codes (both theoretical and practical considerations).

\clearpage
 \onecolumn
\section{Codewords for transversal $ T $ family} \label{appendix:codewords}

{
\centering
\begin{table}[htp]
    \centering
    \begin{tabular}{ll} \toprule
        $((11,2,3))$, $\bmsf{\tau}_3$  &
$\begin{aligned}
\logicalket{0} =   \tfrac{\sqrt{5}}{4} |D_0^{11}\rangle + \tfrac{\sqrt{11}}{4} |D_8^{11}\rangle
\end{aligned}$ \\ \midrule 
     $((27,2,5))$, $\bmsf{\tau}_3$ &
$\begin{aligned}
    \logicalket{0} = & -0.419391 |D_0^{27}\rangle +0.625017 |D_8^{27}\rangle +0.0595089 |D_{16}^{27}\rangle +0.655686 |D_{24}^{27}\rangle
\end{aligned}$ \\ \midrule 
 $((49,2,7))$, $\bmsf{\tau}_1$ &
$\begin{aligned}
    \logicalket{0} = &-0.201735 |D_0^{49}\rangle +0.3284 |D_8^{49}\rangle -0.44571 |D_{16}^{49}\rangle
      +0.517794 |D_{24}^{49}\rangle +0.476193 |D_{32}^{49}\rangle +0.318756
   |D_{40}^{49}\rangle \\
   &+0.237326 |D_{48}^{49}\rangle 
\end{aligned}$ \\ \midrule 
 $((73,2,9))$, $\bmsf{\tau}_4$ &
$\begin{aligned}
    \logicalket{0} =& -0.206212 |D_0^{73}\rangle +0.227626 |D_8^{73}\rangle -0.462933 |D_{16}^{73}\rangle
     +0.352279 |D_{24}^{73}\rangle -0.186197 |D_{32}^{73}\rangle \\
     & -0.383276
   |D_{40}^{73}\rangle
    -0.0735981 |D_{48}^{73}\rangle -0.544932 |D_{56}^{73}\rangle
   +0.156301 |D_{64}^{73}\rangle
   -0.242669 |D_{72}^{73}\rangle 
\end{aligned}$ \\ \midrule 
 $((107,2,11))$, $\bmsf{\tau}_3$ &
$\begin{aligned}
    \logicalket{0} =& 0.241938 |D_0^{107}\rangle -0.430038 |D_8^{107}\rangle +0.0927816
   |D_{16}^{107}\rangle 
   +0.154505 |D_{24}^{107}\rangle +0.429423 |D_{32}^{107}\rangle \\
   &-0.0555468 |D_{40}^{107}\rangle 
    +0.0116979 |D_{48}^{107}\rangle +0.166002
   |D_{56}^{107}\rangle +0.387154 |D_{64}^{107}\rangle 
   +0.0355115
   |D_{72}^{107}\rangle \\
   &-0.116024 |D_{80}^{107}\rangle 
   +0.412203 |D_{88}^{107}\rangle 
+0.102736 |D_{96}^{107}\rangle +0.404712 |D_{104}^{107}\rangle
\end{aligned}$\\ \midrule 
 $((147,2,13))$, $\bmsf{\tau}_2$ &
$\begin{aligned}
    \logicalket{0} =& -0.0655024 |D_0^{147}\rangle -0.13146 |D_8^{147}\rangle -0.0280819
   |D_{16}^{147}\rangle 
   -0.259006 |D_{24}^{147}\rangle +0.0788041
   |D_{32}^{147}\rangle \\
   &-0.218291 |D_{40}^{147}\rangle
   +0.377217 |D_{48}^{147}\rangle
   +0.115604 |D_{56}^{147}\rangle -0.17366 |D_{64}^{147}\rangle
   -0.545858
   |D_{72}^{147}\rangle \\
   &+0.159303 |D_{80}^{147}\rangle +0.23378 |D_{88}^{147}\rangle
   
   +0.344513 |D_{96}^{147}\rangle -0.211129 |D_{104}^{147}\rangle +0.23128
   |D_{112}^{147}\rangle \\
   &-0.13308 |D_{120}^{147}\rangle -0.213004
   |D_{128}^{147}\rangle -0.00407857 |D_{136}^{147}\rangle
   +0.114171
   |D_{144}^{147}\rangle
\end{aligned}$  \\ \bottomrule
    \end{tabular}
    
    \caption{  \Cref{code:codefamily2} for $d=3,5,7,9,11,13$. This family implements a transversal $T$ gate. We have given $\logicalket{0}$, and $\logicalket{1}$ can be found by replacing $\ket*{D^n_w}$ with $\ket*{D^n_{n-w}}$. Finally, $ \bmsf{\tau}_a $ is the relevant irrep of $\Q^{(3)} = \BD_8 = \expval{\X, \T}$.}
    \label{tab:codefamily3numerics}
\end{table}

}

\end{document}